\documentclass[12pt]{article}
\usepackage[a4paper,margin=1in]{geometry}
\usepackage{setspace}
\usepackage[T1]{fontenc}
\usepackage[utf8]{inputenc}
\usepackage{mathpazo}
\usepackage{microtype}
\usepackage[round,authoryear]{natbib}
\usepackage{amsmath,amssymb,amsthm,amsfonts,array,mathrsfs,ifthen,mathtools}
\usepackage{dsfont}
\usepackage{enumitem}
\usepackage{comment}
\usepackage{url}
\usepackage{hyperref}
\usepackage{bm}
\usepackage{graphicx}
\usepackage{xcolor}
\definecolor{warmdarkred}{RGB}{120, 20, 30}
\definecolor{myblue}{RGB}{0,30,180}
\hypersetup{
    colorlinks=true,
    linkcolor=warmdarkred,
    urlcolor=warmdarkred,
    citecolor=warmdarkred
}
\theoremstyle{plain}
\newtheorem{theorem}{Theorem}
\newtheorem{proposition}{Proposition}
\newtheorem{lemma}{Lemma}
\newtheorem{corollary}{Corollary}
\newtheorem{assumption}{Assumption}
\newtheorem{definition}{Definition}

\usepackage[nameinlink,capitalise,noabbrev]{cleveref}
\crefname{assumption}{Assumption}{Assumptions}
\Crefname{assumption}{Assumption}{Assumptions}
\crefname{definition}{Definition}{Definitions}
\Crefname{definition}{Definition}{Definitions}
\crefname{lemma}{Lemma}{Lemmas}
\Crefname{lemma}{Lemma}{Lemmas}
\crefname{proposition}{Proposition}{Propositions}
\Crefname{proposition}{Proposition}{Propositions}
\crefname{corollary}{Corollary}{Corollaries}
\Crefname{corollary}{Corollary}{Corollaries}
\newcommand{\Sellers}{\mathcal S}
\newcommand{\Buyers}{\mathcal B}
\newcommand{\Qualities}{\mathcal Q}
\newcommand{\Types}{\Theta}
\newcommand{\Census}{\mathcal C}
\newcommand{\Domain}{\mathcal D}

\newcommand{\Matchings}{\mathcal X}

\newcommand{\argmax}{\operatorname*{arg\,max}}

\usepackage{titlesec}

\titleformat{\section}
{\normalfont\large\bfseries\centering\sc} 
{\thesection.}                          
{1em}                                   
{}                                      

\titleformat{\subsection}
{\normalfont\normalsize\bfseries\centering\sc}
{\thesubsection.}
{1em}
{}

\title{ \bf \Large
Market tallies: minimal information for efficient trade \\

}
\author{Federico Vaccari\\
        {\small University of Bergamo}\\
        {\small \texttt{vaccari.econ@gmail.com}}}
\date{}

\begin{document}
\maketitle

\begin{abstract}
This paper studies how much public information is needed to implement efficient trade in dynamic markets with privately informed sellers and buyers. An institution compares a certified statistic of market composition
with the statistic implied by agents' reports. Truthful reporting is supported when the statistic changes after every unilateral change in reported type. When all market compositions are possible, the least number of public announcements is $\max\{K,L\}$, where $K$ is the number of seller qualities and $L$ the number of buyer types. The certificate must rely on information outside the reports it checks. The paper also shows that information sufficient to discipline reports need not coordinate buyers across limited capacity. Posted-price implementation may require certified capacities and a clearing rule.
\end{abstract}

\bigskip


\bigskip



\newpage

\begingroup
\hypersetup{linkcolor=black}
\tableofcontents
\endgroup

\newpage


\section{Introduction}
\label{sec:introduction}

Markets with adverse selection often rely on time to produce information. Sellers reveal their quality through their willingness to wait, and buyers learn from prices, offers, and trading histories. This process may improve sorting, but it does so at a cost. Valuable trades are postponed, and some may never occur. When buyers also differ in how much they value quality, learning who should trade is only part of the problem, as the market must also determine who should trade with whom.

This paper asks how much public information is needed to implement efficient trade without delay. The answer is surprisingly modest, provided that the institution can verify a suitable aggregate statistic. The statistic need not reveal individual types, nor need it disclose the full composition of the market. Its purpose is to expose a unilateral change in the reported counts.

To see the idea, consider a finite market in which sellers privately know the quality of their goods and buyers privately know their valuations. Before agents report, an independent source announces a statistic of the numbers of sellers and buyers of each type. The mechanism then calculates the same statistic from the reports. Trade proceeds only if the two values agree. A false report by one participant changes the composition implied by the reports. If the statistic changes with each such alteration, the false report is detected, and the trade is canceled.

The first result characterizes the public statistics that reject every unilateral false report. A statistic has this property if and only if it takes different values at any two market compositions that can be obtained from one another by changing the type of a single participant. When all compositions consistent with the numbers of buyers and sellers are possible, the least number of public values among statistics with this property is $\max\{K,L\}$, where $K$ is the number of seller qualities and $L$ is the number of buyer types. This bound does not depend on the number of participants, even though the number of possible market compositions grows with market size.\footnote{The bound concerns a count check that rejects false reports, whether or not they would be profitable. The check supports truthful reporting as an equilibrium but does not rule out other inefficient equilibria.}

The result follows from a simple counting idea. Index the seller qualities and buyer types, add their indices across participants, and announce the remainder after division by $\max\{K,L\}$. Changing one seller's quality or one buyer's type necessarily changes this remainder. At the same time, fewer public values cannot suffice. Holding all but one seller fixed yields $K$ compositions that must be distinguished, and holding all but one buyer fixed yields $L$ such compositions. Thus, a small public certificate can discipline individual reports even though it reveals little about the underlying market.\footnote{Separation is sufficient for truthful reporting, but need not be necessary when prices or transfers already discourage reports that the tally does not detect. Appendix~\ref{app:beyond_count_checks} studies this broader problem.}


This finding has a natural interpretation in terms of privacy. A platform, registry, or regulator may have access to detailed records without disclosing them to market participants. It can use those records to calculate the public statistics and release only the information required for the report check. Internal verification and public disclosure are distinct. Efficient implementation may require the former without requiring the latter to be equally detailed.

The source of the public statistic is essential. A statistic calculated from the reports that it is meant to check provides no independent information. Agreement then holds by construction. More generally, a procedure based only on the checked reports cannot both accept truthful reports in every state and reject all inconsistent ones. The certificate must rest on an informational source outside those reports. Possible sources include registries, platform records, audits, independent witnesses, and legal certification. The paper makes this institutional requirement explicit and considers several ways in which such information may be produced.

The public statistic addresses incentives, but it does not by itself organize trade. Even if participants know the market's entire composition, several buyers may approach the same seller, while another suitable seller receives no applications. Public information may reveal how many trades should occur without telling buyers how to divide themselves among the available opportunities. Congestion can prevent efficiency even when adverse selection has been resolved.

The paper studies this second problem through a posted-price institution. An organizer uses certified information about market composition to announce prices and capacities. Sellers enter the markets associated with their qualities, buyers apply where their payoff is highest, and a clearing rule assigns buyers when several of them seek the same limited capacity. Prices support the efficient allocation, while clearing prevents buyers from concentrating on the same sellers. This construction requires more information or more institutional involvement than the report check alone.

The comparison separates two functions that are easily conflated. Information used to discipline reports need only expose unilateral inconsistencies. Information used to organize trade must also support prices, identify capacities, and guide participants away from congestion. A public statistic may be sufficient for the first purpose and insufficient for the second. Accordingly, the informational requirements of a market cannot be assessed without specifying the institution in which the information will be used.

The paper's argument applies to finite markets. A false report in a finite market changes an integer count. In a continuum economy, a single participant has measure zero and does not change the aggregate distribution. The results are most applicable to settings in which participation or inventory can be counted, such as thin markets, platform trading rounds, procurement lots, registries, and similar environments.

The paper contributes to the study of adverse selection by treating public information as an input into implementation rather than as a substitute for individual information. It identifies how little needs to be disclosed to support truthful reporting, explains why that disclosure must have an independent source, and shows why additional organization may still be needed to achieve efficient trade. The value of public information depends not only on what it reveals, but also on the task it is asked to perform.


\subsection{Related literature}
\label{sec:related_literature}

This paper builds on the literature on adverse selection that begins with \citet{Akerlof1970}. Much of the subsequent work asks how market interaction itself can reveal private information. Prices, trading histories, entry, and delay may allow goods of different qualities to be sorted over time \citep{JanssenRoy2002,JanssenRoy2004,HendelLizzeriSiniscalchi2005}. Public news and information about supply can also affect the course of trade \citep{DaleyGreen2012,BilanciniBoncinelli2016}, while decentralized dynamic markets may gradually separate informed sellers \citep{MorenoWooders2010,CamargoLester2014}. Buyer heterogeneity introduces a further difficulty because trade must be sorted on both sides of the market \citep{Roy2014}. The present paper asks what independently verified information would allow a market to reach an efficient allocation at once.

The paper is most closely related to work on mechanism design when information about the distribution of types is available. The impossibility result of \citet{MyersonSatterthwaite1983} establishes the tension among efficiency, incentive compatibility, individual rationality, and budget balance in bilateral trade with private information. Aggregate restrictions can alter that tension. \citet{JacksonSonnenschein2007} show that linking many decisions and restricting the frequencies of reports can ease incentive constraints, while \citet{BoukourasKoufopoulos2017} study implementation when realized type frequencies are commonly known. \citet{McLeanPostlewaite2002} relate incentive compatibility to the informational importance of individual agents in large economies. Here, the emphasis is on how little of the realized market composition must be made public to discipline individual reports. The answer is independent of market size when the sets of seller qualities and buyer types are fixed.

There is also a connection with information design. In persuasion and information-design problems, the designer chooses what agents learn in order to shape their beliefs and actions \citep{RayoSegal2010,KamenicaGentzkow2011,BergemannMorris2019}. The public information considered here serves a different purpose. It certifies a feature of the realized market against which individual reports can be checked. The analysis does not rely on a prior or on the management of posterior beliefs. Differently, it asks which distinctions among market compositions are needed for the report check.

This use of public information brings the paper closer to the literature on verifiable disclosure, hard evidence, and certification. Early work studies the disclosure of verifiable private information and the resulting unraveling of information \citep{Grossman1981,Milgrom1981,MilgromRoberts1986}. \citet{GreenLaffont1986} and \citet{BullWatson2007} examine implementation when agents possess evidence or when some claims can be verified. \citet{Lizzeri1999} studies certification intermediaries, and \citet{DranoveJin2010} survey the broader literature on quality disclosure and certification. In the present paper, the certified object is a statistic of the market composition. This difference permits individual information to remain private while still placing a verifiable restriction on reports.

The distinction between verification and disclosure also relates to the economics of privacy and data. \citet{AcquistiTaylorWagman2016} survey the economics of privacy, and \citet{BergemannBonattiGan2022} study the value and governance of market data. This paper does not introduce privacy preferences or a market for data. Its contribution on this front is institutional. An intermediary may use detailed records internally while disclosing only the public information needed for implementation. The amount of information verified by the intermediary may exceed the amount revealed to market participants.

Finally, the analysis of posted prices draws on the assignment-market tradition. Efficient assignments can be supported by prices \citep{ShapleyShubik1971,DemangeGaleSotomayor1986}. The paper uses this logic to construct price markets from certified information about the composition of a finite market. It then adds a clearing rule to allocate limited capacity among buyers. This last step distinguishes the construction from competitive search models, in which prices and submarkets screen agents through decentralized choice \citep{GuerrieriShimerWright2010}. In a finite market, prices may direct buyers toward suitable goods without preventing several buyers from choosing the same seller. The resulting congestion is the reason that information and clearing play separate roles in the analysis.


\section{Model}
\label{sec:model}

There are $N\geq 1$ sellers and $M\geq 1$ buyers. Let $\Sellers=\{1,\ldots,N\}$ and $\Buyers=\{1,\ldots,M\}$ denote the two sets of agents. Each seller owns one indivisible good, and each buyer demands at most one good.

A seller's private information concerns the quality of her good. There are $K\geq 1$ possible qualities, $\Qualities=\{q^0,\ldots,q^{K-1}\}$, and seller $i$ observes $q_i\in\Qualities$. A buyer's private information is her valuation type. There are $L\geq 1$ possible buyer types, $\Types=\{\theta^0,\ldots,\theta^{L-1}\}$, and buyer $j$ observes $\theta_j\in\Types$. Write $q=(q_i)_{i\in\Sellers}$ and $\theta=(\theta_j)_{j\in\Buyers}$ for the realized profile of qualities and buyer types.

A buyer of type $\theta$ values a good of quality $q$ at $v_\theta(q)$. The seller's value from retaining a good of quality $q$ is $c(q)$. Both functions are common knowledge and take finite real values. No monotonicity or single-crossing assumption is imposed unless stated otherwise.

If buyer $j$ purchases seller $i$'s good at price $p$ and date $t$, their payoffs, measured relative to their outside options, are
\[
        \delta^t\left(v_{\theta_j}(q_i)-p\right)
\]
and
\[      
        \delta^t\left(p-c(q_i)\right),
\]
respectively, where $\delta\in(0,1)$. An agent who does not trade receives zero. Time is discrete, with $t\in\{0,1,2,\ldots\}$. The institutions considered below seek to carry out all equilibrium trades at date zero. Discounting is relevant to the cost of delay, but not to the comparison of date-zero allocations.

A matching $x\subseteq\Sellers\times\Buyers$ is \emph{feasible} if each agent belongs to at most one matched pair. Let $\Matchings$ be the set of feasible matchings. At the profile $(q,\theta)$, matching $x$ produces total gains from trade
\[
        W(x\mid q,\theta) = \sum_{(i,j)\in x} \left[v_{\theta_j}(q_i)-c(q_i)\right].
\]
The first-best surplus is
\[
        W^*(q,\theta) = \max_{x\in\Matchings}W(x\mid q,\theta).
\]
The empty matching is feasible, so $W^*(q,\theta)\geq 0$.

Some profiles may admit more than one surplus-maximizing matching. When a mechanism requires a single selected outcome, fix a public tie-breaking rule and let $x^*(q,\theta)$ denote the matching it selects. The rule first removes any zero-surplus trade and then chooses among the remaining surplus-maximizing matchings. Every trade in $x^*(q,\theta)$ produces positive gains.

The choice among surplus-maximizing matchings has no welfare significance. Throughout the paper, a matching is \emph{efficient} if it attains $W^*(q,\theta)$, and \emph{efficient trade} means that such a matching is carried out at date zero. The count-check mechanism uses $x^*(q,\theta)$ because it assigns a particular matching. A market-clearing institution may select another efficient matching when the first best is not unique.


\subsection{Market composition and public tallies}
\label{subsec:censuses_tallies}

The composition of the market can be described without revealing the identity of any participant. For each seller quality $q^r$, and for $r\in\{0,\ldots,K-1\}$, let
\[
        n_r(q) = \#\left\{i\in\Sellers \mid q_i=q^r \right\},
\]
and, for each buyer type $\theta^\ell$ and $\ell\in\{0,\ldots,L-1\}$, let
\[
        m_\ell(\theta) = \#\left\{j\in\Buyers\mid \theta_j=\theta^\ell\right\}.
\]

\begin{definition}
The \textbf{market census} at profile $(q,\theta)$ is
\[
        C(q,\theta) = \left( (n_r(q))_{r=0}^{K-1}, (m_\ell(\theta))_{\ell=0}^{L-1} \right).
\]
The market census records the number of sellers of each quality and the number of buyers of each type, but not their identities.
\end{definition}

The set of all censuses consistent with the numbers of sellers and buyers is
\[
        \Census_{N,M} = \left\{ (n,m)\in\mathbb Z_+^K\times\mathbb Z_+^L \;\text{ such that }\; \sum_{r=0}^{K-1}n_r=N,\ \sum_{\ell=0}^{L-1}m_\ell=M \right\}.
\]
The set of market compositions that may occur is a nonempty domain $\Domain\subseteq\Census_{N,M}$.

The omission of identities is substantive. The paper studies aggregate certification rather than public verification of each participant's type. Individual identities have no direct payoff relevance in the model, while publishing the association between identities and types would disclose far more information than is needed for the count check. A registry or platform may hold identity-linked records internally in order to certify the census, but the public census contains only type counts.

Allowing identity-linked public information would lead to a different implementation problem. If the type of every named participant were certified, the mechanism could use that information directly rather than elicit the same types through reports. Even a more limited identity-sensitive certificate could distinguish exchanges of reports that leave the census unchanged. The results below characterize what can be achieved with aggregate public information. However, the same limitations do not necessarily apply to individual verification. An anonymous census need not prevent all inference about individuals, particularly in a small or restricted market. It means only that the public record does not explicitly associate a type with an identity.

\begin{definition}
The domain is \textbf{unrestricted} if $\Domain=\Census_{N,M}$. That is, every composition consistent with the fixed numbers of sellers and buyers may occur. 

A \textbf{restricted domain} is a proper nonempty subset $\Domain\subsetneq\Census_{N,M}$, and represents prior knowledge that some compositions cannot occur.
\end{definition}

A reported profile whose census does not belong to $\Domain$ is regarded as infeasible and is rejected by the mechanisms considered below.

The designer need not disclose the market census itself. Instead, the public announcement may retain only some of its information.

\begin{definition}
A \textbf{market tally} is a function $\phi:\Domain\to\mathcal Z$, where $\mathcal Z$ is the finite set of announcements that the tally may produce. Without loss, $\phi$ is taken to be onto, so every element of $\mathcal Z$ is used for at least one census. The tally reveals the \textbf{full market census} if $\phi$ is one-to-one.
\end{definition}

In the count-check mechanism introduced below, the realized tally is an authenticated public announcement. Both the trading institution and market participants observe it before agents report. The institution uses the tally to check the reports, while the records or other evidence used to produce it remain private. A certificate sent only to a trusted institution would describe a different arrangement. The disclosure results below concern the public-announcement benchmark.

A tally that is not one-to-one allows several market compositions to produce the same public announcement. The number of possible announcements, $|\mathcal Z|$, will be the paper's measure of public disclosure. It counts the authenticated messages that may be released to market participants. It does not measure how much information the certifying institution collects or holds, or how costly that information is to verify. In particular, a tally with a small range may still be computed from the full census or from identity-linked records. Moreover, the measure is prior-free, as it records how many announcements may be needed, not how frequently they arise.

The analysis first asks how small $|\mathcal Z|$ can be when the tally is used to compare individual reports with independently certified aggregate information. This is a question about the count-check institution introduced in the next section. It does not presume that the same amount of information is sufficient for every possible trading institution.


\section{Analysis}
\label{sec:analysis}


\subsection{Certified count checks}
\label{sec:certified_count_checks}

Suppose that an informational anchor has access to records or other evidence sufficient to certify the tally of the true market census. Before agents report, it releases the authenticated value
\[
        z_0=\phi(C(q,\theta)).
\]
The trading institution and market participants observe $z_0$, but they need not observe the underlying records or the full census. The institution then compares $z_0$ with the tally implied by agents' reports. For now, the certificate is assumed to be correct. Section~\ref{sec:certification} explains why its content cannot be derived solely from the reports being checked.

A false report by one seller reduces the reported number of sellers of her true quality by one and increases the reported number of sellers of another quality by one. A false report by one buyer has the analogous effect on the buyer counts. This observation motivates the following relation between censuses.

\begin{definition}
Two distinct censuses $C,C'\in\Domain$ are \textbf{one-agent neighbors}, written $C\sim_1 C'$, if one can be obtained from the other by changing the type of a single seller or a single buyer. Thus, either one unit is moved between two seller-quality counts while all buyer counts remain fixed, or one unit is moved between two buyer-type counts while all seller counts remain fixed.
\end{definition}

The required public information can now be stated simply.

\begin{definition}
A tally $\phi:\Domain\to\mathcal Z$ is \textbf{separating} on $\Domain$ if, for every $C,C'\in\Domain$,
\[
        C\sim_1 C' \implies\phi(C)\neq\phi(C').
\]

\end{definition}

Separation does not require the tally to reveal the census. It requires only that the public announcement change when the reported census is altered by one participant.

Fix a tally $\phi:\Domain\to\mathcal Z$. The associated \emph{count-check mechanism} proceeds as follows.

\begin{enumerate}[label=\textbf{\arabic*.},leftmargin=2em]
    \item The certified tally $z_0=\phi(C(q,\theta))$ is announced;

    \item Each seller reports a quality $\hat q_i\in\Qualities$, and each buyer reports a valuation type $\hat\theta_j\in\Types$. These reports produce the census $\widehat C=C(\hat q,\hat\theta)$;

    \item If $\widehat C\notin\Domain$ or $\phi(\widehat C)\neq z_0$, the mechanism cancels trade and makes no transfers;

    \item If $\widehat C\in\Domain$ and $\phi(\widehat C)=z_0$, the mechanism selects $x^*(\hat q,\hat\theta)$. For every matched pair $(i,j)$, buyer $j$ pays seller $i$
    \[
            p_{ij}(\hat q,\hat\theta) = \frac{ v_{\hat\theta_j}(\hat q_i)+c(\hat q_i) }{2}.
    \]
    Unmatched agents make and receive no transfers.
\end{enumerate}

The transfer divides the reported gains from trade equally between the buyer and the seller. Since the mechanism selects only trades with positive reported surplus, both parties to a truthful trade receive a positive payoff. Transfers are balanced within each pair.\footnote{Equal division is convenient but not essential. The argument would also apply to any rule that places the transaction price between reported cost and reported value and gives both parties a positive share whenever the reported gains are positive.}

The next result shows that separation supports truthful reporting, characterizes the stronger requirement that every unilateral false report be rejected, and determines the least public information needed for that requirement.

\begin{theorem}
\label{thm:main_countcheck}
Let $\phi:\Domain\to\mathcal Z$ be a market tally. Then,

\begin{enumerate}[label=\textup{(\roman*)},leftmargin=2em]
    \item If $\phi$ is separating, truthful reporting is an ex post equilibrium of the count-check mechanism. At every realized profile whose census belongs to $\Domain$, the truthful equilibrium carries out $x^*(q,\theta)$ at date zero, satisfies ex post individual rationality, and balances transfers;

    \item The tally rejects every unilateral false report at every census in $\Domain$ if and only if it is separating;

    \item On the unrestricted domain, the least number of possible public announcements among separating tallies is
    \[
            R\coloneqq \max\{K,L\}.
    \]
    A tally attaining this bound is
    \begin{equation}
            \phi(C) = \left( \sum_{r=0}^{K-1}r n_r + \sum_{\ell=0}^{L-1}\ell m_\ell \right)\bmod R,
            \label{eq:modular_tally}
    \end{equation}
    where $C= \left( (n_r)_{r=0}^{K-1}, (m_\ell)_{\ell=0}^{L-1} \right)$.
\end{enumerate}
\end{theorem}

In \eqref{eq:modular_tally}, ``$\bmod R$'' means the remainder after division by $R$, taking values in $\{0,\ldots,R-1\}$. A tally constructed from such remainders is called a \emph{modular tally}. It acts as a check on the reported type counts.

The incentive argument is immediate. When all other agents report truthfully, a false report either produces a census outside $\Domain$ or moves the census to a one-agent neighbor. In the first case the report is rejected as infeasible. In the second, separation makes its implied tally differ from the certified announcement. The deviation leads to no trade. Truthful reporting gives every agent a nonnegative payoff and gives each matched agent a positive payoff.

The bound on public announcements has a similarly straightforward interpretation. Holding all buyers and all but one seller fixed produces $K$ censuses, one for each possible quality of the remaining seller. Every two of these censuses are one-agent neighbors and must receive different announcements. At least $K$ announcements are required. Repeating the argument with one buyer and $L$ possible buyer types gives a lower bound of $L$.

The construction in \eqref{eq:modular_tally} reaches the larger of these two bounds. If one seller changes her reported quality from $q^r$ to $q^s$, the quantity inside parentheses changes by $s-r$. Since $0<|s-r|<R$, its remainder after division by $R$ changes. The same reasoning applies when a buyer changes her report from $\theta^\ell$ to $\theta^h$. Thus, the tally detects every unilateral change while using only $R$ announcements. Figure~\ref{fig:three_announcement_tally} illustrates this construction for three seller qualities, holding the buyer census fixed.

\begin{figure}[t]
    \centering
    \includegraphics[width=0.55\textwidth]{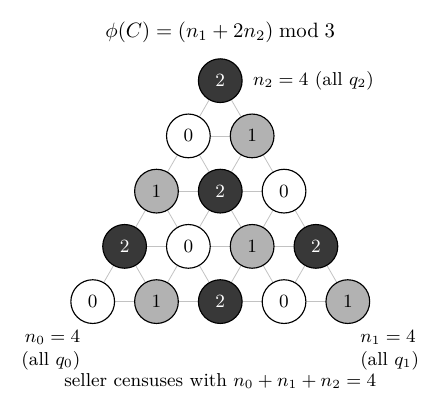}
    \caption{\textit{A three-announcement tally.} Each point represents a seller census with four sellers and three possible qualities. An edge joins two censuses that differ in one seller's report. The number inside each point is the announcement $\phi(C)=(n_1+2n_2)\bmod 3$. Every edge joins points with different announcements, although many other censuses remain pooled.}
    \label{fig:three_announcement_tally}
\end{figure}

For a general domain, let
\[
        \kappa^{\mathrm{cc}}(\Domain) \coloneqq \min_{\phi\ \mathrm{separating\ on}\ \Domain} |\phi(\Domain)|
\]
denote the least number of announcements required by a separating tally. The theorem gives
\[
        \kappa^{\mathrm{cc}}(\Census_{N,M}) = \max\{K,L\}.
\]
Restrictions on the domain may reduce this number because they exclude some censuses to which a unilateral false report might otherwise lead. If $\Domain$ contains no pair of one-agent neighbors, one announcement is enough, as every unilateral change either produces an infeasible census or is already ruled out by the domain.

Theorem~\ref{thm:main_countcheck} separates two claims. Part~(i) uses separation as a sufficient condition for truthful ex post equilibrium. Parts~(ii) and~(iii) concern the stronger requirement that every unilateral false report fail the count check, whether or not it would have been profitable. If two neighboring censuses receive the same announcement, a false report can pass the check, but it need not benefit the deviating agent. Accordingly, $\kappa^{\mathrm{cc}}(\Domain)$ is the least disclosure required for this payoff-independent count check. However, it is not necessarily the least disclosure capable of supporting truthful behavior under the particular preferences and transfers of an economic environment.\footnote{Appendix~\ref{app:beyond_count_checks} studies this latter problem.}

The theorem establishes weak, rather than full, implementation. Starting from truthful reports, a unilateral false report makes the mechanism cancel trade. An agent who would trade truthfully strictly prefers to avoid that outcome, but an unmatched agent may be indifferent. More importantly, the tally verifies counts rather than identities. It cannot detect an exchange of reports that leaves the census unchanged. The mechanism may also remain at no trade when reports fail the check and no single agent can restore consistency.\footnote{Appendix~\ref{sec:weak_full_implementation} characterizes both accepted false equilibria and rejected no-trade equilibria. Ruling them out requires an additional instrument, such as identity-linked verification, audits, penalties, or a procedure for revising rejected reports.}

A tally with two possible announcements will be called a \emph{one-bit tally}. It works like an on-and-off switch, and is represented by Figure~\ref{fig:one_bit_tally}. The leading special case is immediate.

\begin{corollary}
\label{cor:binary}
Suppose there are two seller qualities and two buyer types. On the unrestricted domain, two announcements are necessary and sufficient for unilateral count checking. One such tally is $\phi(C)=(n_H+m_H)\bmod 2$, where $n_H$ is the number of high-quality sellers and $m_H$ is the number of high-valuation buyers.

More generally, on any domain $\Domain$,
\[
        \kappa^{\mathrm{cc}}(\Domain) =
        \begin{cases}
        2, & \text{if $\Domain$ contains a pair of one-agent neighbors},\\ 1, & \text{otherwise}.
        \end{cases}
\]
\end{corollary}

The announcement reveals only whether $n_H+m_H$ is even or odd, and reveals neither count separately. Yet, a false report by one participant changes one of the two counts by one and changes the announcement accordingly. Figure~\ref{fig:one_bit_tally} illustrates the construction. Neighboring censuses always receive different announcements, although many distant censuses remain pooled.

\begin{figure}[t]
    \centering
    \includegraphics[width=0.65\textwidth]{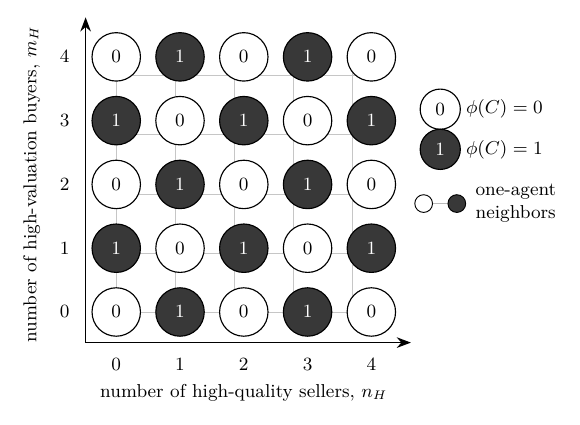}
    \caption{\textit{A one-bit tally in a binary market.} Each node represents a census $(n_H,m_H)$. White and dark nodes correspond to the two values of $\phi(C)=(n_H+m_H)\bmod 2$. A unilateral change in a report moves to a horizontal or vertical neighbor and changes the public announcement. The tally nevertheless pools many censuses that the full census would distinguish.}
    \label{fig:one_bit_tally}
\end{figure}


\subsection{Certification and informational anchors}
\label{sec:certification}

Theorem~\ref{thm:main_countcheck} takes the certified tally as given. This assumption is essential. A public announcement can discipline reports only if its content is not determined by the reports being checked. If the institution calculates the tally from those same reports, agreement holds by construction.

It is useful to state this point formally. Write $\omega=(q,\theta)$ for a profile of seller qualities and buyer types, and let
\[
        \Omega_{\Domain} \coloneqq \left\{ \omega\in\Qualities^\Sellers\times\Types^\Buyers \; \text{ such that }  \; C(\omega)\in\Domain \right\}
\]
be the set of profiles admitted by the domain. A certification procedure is \emph{report-only} if its decision depends exclusively on messages submitted by the agents. Thus, if $M_a$ is the message space of agent $a$, its acceptance rule has the form
\[
        A: \prod_{a\in\Sellers\cup\Buyers}M_a \to  \{0,1\}.
\]
The procedure has no access to a registry, audit, witness, platform record, or other information about the realized profile. Every message is available to every type of the agent who sends it.

For each $\omega\in\Omega_{\Domain}$, let $m^T(\omega)$ denote the prescribed message profile that truthfully reports the agents' types and the tally $\phi(C(\omega))$. At the true profile $\omega$, a reported profile $\hat\omega$ is \emph{tally-inconsistent} if $\phi(C(\hat\omega)) \neq \phi(C(\omega))$. An independently certified tally rejects such a report because its implied tally differs from the public certificate. A report-only procedure cannot do so in every state while also accepting truthful reports in every state.

\begin{proposition}
\label{prop:no_self_certification_main}
Suppose that $\phi$ is nonconstant on $\Domain$. No report-only certification procedure can both accept the prescribed truthful message profile $m^T(\omega)$ at every $\omega\in\Omega_{\Domain}$ and reject every tally-inconsistent report.
\end{proposition}

The proposition rests on observational equivalence. To a designer who sees only reports, a truthful description of $\omega'$ when $\omega'$ has occurred looks no different from the same description submitted when the true profile is $\omega$. The messages alone provide no basis for treating the two cases differently. The revelation principle does not alter this conclusion. It allows one to simplify the agents' messages after the information available to the mechanism has been specified. It does not supply information about the realized market that the mechanism does not otherwise observe. Certification is a restriction on the information structure, not merely on the form of the reporting game.

The most direct failure of report-only procedures occurs when the institution constructs the public tally from the reports themselves.

\begin{corollary}
\label{cor:self_generated_vacuous_main}
Suppose that agents report a profile $\hat\omega\in\Omega_{\Domain}$ and that the institution announces $\hat z=\phi(C(\hat\omega))$. If the report is accepted whenever $\phi(C(\hat\omega))=\hat z$, every reported profile whose census belongs to $\Domain$ is accepted. Any rejection of a census outside $\Domain$ follows from the domain restriction, not from the tally comparison.
\end{corollary}

An \emph{informational anchor} is a source of evidence about the realized market that is independent of the reports being checked. The anchor computes or certifies the tally and releases an authenticated value. It may use the full census or identity-linked records internally, while disclosing only the tally. However, the anchor and the trading institution need not be separate organizations. A platform, for example, may perform both roles. What matters is the separation between the records used to certify the tally and the reports that the tally checks. Theorem~\ref{thm:main_countcheck} limits the public announcement, not the information that the institution must collect or verify.

The subsequent proposition states only that some information must enter from outside the unrestricted reports. Appendix~\ref{app:certification_extensions} considers two possible arrangements: certification by independent witnesses and certification supported by audits.


\subsection{Certified posted-price clearing}
\label{sec:posted_prices}

The count-check mechanism asks agents to report their types and then assigns buyers directly to sellers. This section considers a more decentralized institution. Sellers choose among posted-price markets, buyers apply to those markets, and an organizer clears the resulting applications. The organizer does not assign agents on the basis of individual type reports, but she still certifies the available markets and allocates their limited capacity.

This change of institution changes the use of information. The count check needs enough information to expose an individual inconsistency. Posted-price trade must also determine which markets should open, how many sellers should enter each one, and how buyers should be allocated when several of them seek the same capacity.

The construction requires a certified census, but the census need not be observed by everyone. Specifically, the theorem below uses the following arrangement. The organizer receives the certified census but does not publish it. She uses the seller and buyer counts to calculate prices and capacities, and announces only the price markets, seller capacities, and clearing rule. Because the capacities equal the seller counts, the seller side of the census becomes public. The buyer census need not be announced. Full public disclosure of the census is a sufficient special case.

Fix a census $C$. For each quality $q\in\Qualities$, let $n_q$ be the number of goods of that quality, and, for each type $\theta\in\Types$, let $m_\theta$ be the number of such buyers. Matching a type-$\theta$ buyer with a quality-$q$ good produces surplus
\[
        s_{\theta q}=v_\theta(q)-c(q).
\]

A \emph{type-class allocation} is a collection of nonnegative integers
\[
        y=(y_{\theta q})_{\theta\in\Types,q\in\Qualities},
\]
where $y_{\theta q}$ is the number of type-$\theta$ buyers assigned to goods of quality $q$. It is feasible if
\[
        \sum_{\theta\in\Types}y_{\theta q}\leq n_q \quad\text{for every }q\in\Qualities
\]
and
\[
        \sum_{q\in\Qualities}y_{\theta q}\leq m_\theta \quad\text{for every }\theta\in\Types.
\]
It is efficient if it maximizes
\[
        \sum_{\theta\in\Types} \sum_{q\in\Qualities} y_{\theta q}s_{\theta q}
\]
over all feasible type-class allocations.

The following standard assignment result supplies prices that support an efficient allocation.

\begin{lemma}
\label{lem:support}
For every census $C$, there are an efficient type-class allocation $y^*$ and nonnegative numbers $(u_\theta)_{\theta\in\Types}$ and $(\rho_q)_{q\in\Qualities}$ such that, for every $(\theta,q)\in\Types\times\Qualities$,
\[
        u_\theta+\rho_q\geq s_{\theta q} = v_\theta(q)-c(q),
\]
with equality whenever $y^*_{\theta q}>0$. Moreover,
\[
        \sum_{q\in\Qualities}y^*_{\theta q}<m_\theta \implies u_\theta=0
\]
and
\[
        \sum_{\theta\in\Types}y^*_{\theta q}<n_q \implies \rho_q=0.
\]

At prices $p_q=c(q)+\rho_q$, every buyer type assigned to quality $q$ obtains
\[
        v_\theta(q)-p_q = u_\theta = \max\left\{ 0,\max_{r\in\Qualities} \left[v_\theta(r)-p_r\right] \right\}.
\]
A seller who trades in market $q$ receives the nonnegative payoff $\rho_q$.
\end{lemma}

The numbers $u_\theta$ and $\rho_q$ divide the surplus supported by the assignment prices. If a type-$\theta$ buyer is assigned to quality $q$, then $v_\theta(q)-p_q=u_\theta$. No other quality gives that buyer more than $u_\theta$. Positive $\rho_q$ means that every good of quality $q$ must be sold in the supported allocation. Likewise, positive $u_\theta$ means that every buyer of type $\theta$ must be served. These two observations determine the clearing requirements below.

For every quality $q$, the organizer announces a \emph{price market} carrying the label $q$ and the transaction price $p_q$. The quality label matters when two qualities have the same price, as it preserves them as separate certified markets with separate capacities.

Each seller chooses one price market. Let
\[
        \widetilde n_q \coloneqq \#\{ i\in\Sellers \mid  \text{seller }i\text{ chooses market }q \}
\]
be the number of sellers choosing market $q$. The vector $(\widetilde n_q)_{q\in\Qualities}$ records sellers' market choices, while $(n_q)_{q\in\Qualities}$ gives the certified number of sellers of each quality.

For buyers, define
\[
        U_\theta(q) \coloneqq v_\theta(q)-p_q
\]
and
\[
        U_\theta^* \coloneqq \max\left\{ 0,\max_{q\in\Qualities}U_\theta(q) \right\}.
\]
Thus, $U_\theta^*$ is the highest payoff available to a type-$\theta$ buyer, including the outside option.

Each buyer application carries one of two designations. Under the intended strategy, a buyer designates an application to market $q$ as \emph{strong} if $U_\theta(q)=U_\theta^*>0$, and as \emph{standby} if $U_\theta(q)=U_\theta^*=0$. A strong application seeks a market that gives the buyer her highest available payoff and a strict gain from trade. A standby application records willingness to trade when the buyer is indifferent between trading and remaining unmatched. These designations are messages submitted by buyers, but not certificates of buyer type.

The posted-price clearing protocol proceeds as follows.
\begin{enumerate}[label=\textbf{\arabic*.},leftmargin=2em]
    \item The organizer announces the price markets $(q,p_q)_{q\in\Qualities}$ and seller capacities $(n_q)_{q\in\Qualities}$;

    \item Each seller chooses one price market. Trade proceeds only if $\widetilde n_q=n_q$ for every $q\in\Qualities$. If this condition fails, the market closes and no trade occurs;

    \item If the market opens, each buyer submits a finite list of applications, designating each one as either \emph{strong} or \emph{standby};

    \item The organizer seeks a feasible matching with the following properties:
    \begin{enumerate}[label=\textup{(\roman*)}]
        \item every buyer who submitted a strong application is assigned to a market to which she applied strongly;

        \item every seller in a market with $p_q>c(q)$ is assigned to a buyer who submitted either a strong or a standby application to that market;

        \item no buyer is assigned to a market to which she did not apply.
    \end{enumerate}
    If several such matchings exist, the organizer first keeps those with the smallest number of trades and then uses a fixed public tie-breaking rule. If none exists, no trade occurs. Every matched buyer pays the price of the market to which she is assigned.
\end{enumerate}

The intended seller strategy is to choose the market bearing the seller's true quality label. The intended buyer strategy is as follows. If $U_\theta^*>0$, the buyer applies strongly to every market that gives her $U_\theta^*$ and submits no other application. If $U_\theta^*=0$, she submits standby applications to every market that gives her zero payoff. She does not apply to a market that gives a negative payoff.

The minimum-trade convention matters only when both sides of a possible trade receive zero. A standby buyer who trades in a market with $p_q=c(q)$ receives zero, as does the seller, and the trade creates no surplus. The clearing rule omits such trades.

\begin{theorem}
\label{thm:posted}
Suppose that the organizer observes a correctly certified market census and announces the prices and seller capacities constructed above. Then,
\begin{enumerate}[label=(\roman*)]
    \item The prescribed seller choices and buyer applications form an ex post equilibrium of the posted-price clearing protocol;
    \item  The full census itself need not be disclosed to market participants;
    \item At every realized type profile, all equilibrium trades occur at date zero, the resulting matching maximizes total surplus and contains no zero-surplus trade, and every participant receives a nonnegative payoff;
    \item For any common prior over type profiles and any commonly known rule for selecting supporting prices at each census, the same strategies can be completed with beliefs to form a pure-strategy perfect Bayesian equilibrium.
\end{enumerate}
\end{theorem}

The role of the clearing rule is worth emphasizing. At the prices supplied by Lemma~\ref{lem:support}, buyers with positive gains apply only to markets that give them their highest payoff. This does not ensure that their independent applications respect capacity. The clearing rule chooses among their acceptable markets so that all buyers with positive supported payoffs are served and all sellers with positive supported payoffs trade. Any matching with these properties attains the first-best surplus.

The theorem separates the information needed to construct the market from the information disclosed to participants. Lemma~\ref{lem:support} uses both the seller and buyer counts to calculate an efficient type-class allocation and supporting prices. The organizer must have access to those counts in this construction. Market participants, however, need observe only the resulting prices, seller capacities, and clearing rule.

For the next result, consider this alternative information arrangement. The organizer has no aggregate information beyond a public signal of the census, $\psi(C)$. Prices and capacities must be determined from that signal, which must contain enough information to support the announcements made by the organizer. The result gives a necessary condition for announcing the correct seller capacities. However, it is not a complete characterization of the information needed to calculate supporting prices.

Let $\psi:\Domain\to\mathcal Z$ be a public signal. After observing $z=\psi(C)$, the organizer announces a seller-capacity vector $\overline n(z) = (\overline n_q(z))_{q\in\Qualities}$, where $\sum_{q\in\Qualities}\overline n_q(z)=N$. The market bearing label $q$ is intended for sellers of quality $q$, and trade proceeds only when the number of sellers choosing each market equals the announced capacity of that market.

\begin{proposition}
\label{prop:seller_capacity_lower_bound}
Suppose that $\psi(C)$ is the organizer's only aggregate information about the seller side of the market. If a quality-labelled posted-price protocol uses the capacity check described above and implements trade at every $C\in\Domain$, then
\[
n(C)\neq n(C') \implies \psi(C)\neq\psi(C')
\]
for every $C,C'\in\Domain$, where $n(C)=(n_q(C))_{q\in\Qualities}$ is the seller count vector.

On the unrestricted domain, the public signal must have at least
\[
\binom{N+K-1}{K-1}
\]
possible values.
\end{proposition}

Proposition~\ref{prop:seller_capacity_lower_bound} gives a seller-side requirement. The result does not imply that the full market census must be public, nor does it characterize all the information needed to calculate prices. If the organizer privately observes the census, she can calculate prices and capacities before announcing the resulting market structure.

If the public signal is the organizer's only aggregate information, the signal must also permit her to choose suitable prices. Lemma~\ref{lem:support} shows that supporting prices may depend on both seller and buyer composition. However, this does not mean that the signal must always reveal the buyer census, as the same prices may support efficient choices at several buyer censuses. Appendix~\ref{subsec:buyer_self_selection} gives sufficient conditions for that possibility.

The tally in Theorem~\ref{thm:main_countcheck} generally does not contain enough information to run the posted-price protocol, as it was designed to expose a unilateral inconsistency, but not to recover seller capacities or calculate prices. Report checking and market clearing place different demands on information.


\subsection{Information for incentives and for coordination}
\label{sec:incentives_coordination}

The preceding results give public information two different roles. In the count-check mechanism, the tally tests whether individual reports are consistent with the composition of the market. In the posted-price protocol, information is also used to determine prices and capacities, while the clearing rule allocates buyers among the available sellers. The first task is one of incentives, whereas the second is one of coordination.

The distinction can be seen in a one-shot buyer-choice institution. At date zero, sellers are publicly available on announced terms. After observing the public information, each buyer simultaneously chooses one seller or the outside option. Buyers receive no private recommendations. If exactly one buyer chooses a seller, they trade. If several buyers choose the same seller, at most one trades, according to a fixed tie-breaking rule, and the others remain unmatched. An unmatched buyer cannot turn to another seller within the same trading round.

I will refer to \emph{congestion} as the event in which several buyers choose the same seller while another seller who could support a valuable trade remains unused. The institution ends after the date-zero choices. One could allow an unsuccessful buyer to search again at a later date, but this would not restore immediate efficiency, as it would reduce surplus because $\delta<1$. The one-shot formulation isolates the coordination problem before delay is allowed to correct it.

Figure~\ref{fig:anonymous_coordination} gives the core of the coordination problem. The census permits three different pairs of high-valuation buyers, but there are only two sellers. Whatever seller choices are assigned to high-valuation buyers, some possible pair must be assigned to the same seller. Randomization can alter the probability of congestion, but it cannot rule it out at every profile consistent with the census.

\begin{figure}[t]
    \centering
    \includegraphics[width=0.8\textwidth]{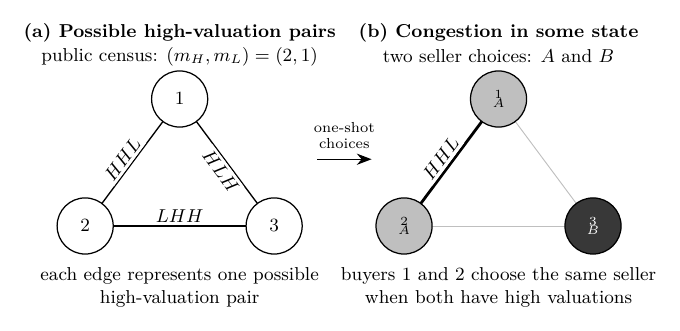}
    \caption{\textit{Why an anonymous census does not coordinate buyers.} Panel (a) shows the three possible high-valuation pairs consistent with a census containing two high-valuation buyers and one low-valuation buyer. In panel (b), the letter inside each node denotes the seller chosen by that buyer when she has a high valuation. With three potential high-valuation buyers and only two sellers, some pair must choose the same seller. If that pair has high valuations, congestion leaves the other seller unused.}
    \label{fig:anonymous_coordination}
\end{figure}

\begin{proposition}
\label{prop:incentive_coordination_main}
There is a finite market in which a one-bit tally is sufficient for unilateral count checking, but even the full market census does not permit the one-shot buyer-choice institution to guarantee immediate efficient trade at every type profile consistent with that census.
\end{proposition}

The example on which Proposition~\ref{prop:incentive_coordination_main} is based on has two identical sellers and three buyers. The public census reveals that two buyers have high valuations and that both goods should be sold to them. After learning her own type, a high-type buyer still cannot identify which of the other two buyers is also a high type. She must choose a seller using only her own type, her identity, and the information common to the market. A low-type buyer can infer that the other two buyers are high types, but she takes the outside option and does not resolve their coordination problem.

This is not enough to avoid congestion in every state. Any rule that assigns a seller to each of the three possible high-type buyer identities must direct at least two identities to the same seller. If those two buyers are the high types, they meet at one seller while the other seller remains unused. Randomization cannot remove this possibility for every pair of high types. The difficulty here lies on the absence of information or instructions that divide the relevant buyers between the two sellers.

A clearing rule provides those instructions after observing the applications. It can assign buyers across markets while respecting the choices they find acceptable. A private recommendation based on the realized applications, a queue, or a direct assignment could perform a similar function. If unsuccessful buyers are allowed to search again, the market uses delay to repair the initial congestion rather than achieving efficiency at date zero.

The example underlying the above proposition does not imply that decentralized choice must always fail. Prices may sometimes separate buyers in a way that also respects the available capacities. Appendix~\ref{app:decentralized_prices} gives conditions under which this is possible. The proposition establishes only that aggregate disclosure, even when it reveals the full census, does not generally supply the coordination required for immediate efficient trade.


\subsection{Economizing on public disclosure}
\label{sec:privacy_disclosure}

The count-check result separates the information an institution may need to verify from the information it must reveal publicly. An organizer may observe the census in order to certify the tally, but the tally itself can disclose far less than the census.

One way to measure the difference between tallies and censuses is that of counting the number of announcements that may occur. On the unrestricted domain, the seller side of the census can take
\[
        \binom{N+K-1}{K-1}
\]
different values, while the buyer side can take
\[
        \binom{M+L-1}{L-1}
\]
different values. Full disclosure must distinguish every combination of these two count vectors.

\begin{lemma}
\label{prop:disclosure_gap_main}
On the unrestricted domain, disclosure of the full market census requires
\begin{equation}
        |\Census_{N,M}| = \binom{N+K-1}{K-1} \binom{M+L-1}{L-1}
        \label{eq:full_census_messages}
\end{equation}
possible public announcements. A tally sufficient to reject every unilateral false report requires only $\max\{K,L\}$ announcements. For fixed $K$ and $L$, the number required by the tally does not vary with market size. The number of possible censuses grows with $N$ whenever $K\geq2$ and with $M$ whenever $L\geq2$.
\end{lemma}

The comparison concerns public disclosure, not the information collected by the organizer. Producing the tally may require access to individual records or to the full census. The lemma says that these records need not all be revealed to market participants. An intermediary can use detailed information for verification while publishing only the announcement required for the count check.

Nor should the result be read as a general privacy guarantee. A tally still reveals information about the market, and its implications may depend on what participants already know. Furthermore, the lemma uses a prior-free criterion, as it counts the number of announcements that the institution may have to make. By constrast, it does not model privacy preferences, limit statistical inference, or protect individual information under every possible prior.

If a prior over market compositions is available, disclosure can instead be measured by the entropy of the public announcement. The institution would then choose, among tallies that distinguish one-agent neighbors, the one whose announcements carry the least expected information under that prior. This criterion may favor an uneven use of announcements, assigning a common announcement to likely censuses whenever the incentive requirement permits. Appendix~\ref{subsec:privacy_objective} develops this prior-dependent formulation.

The two criteria answer different questions. Counting announcements asks how many public outcomes may be required in the worst case. Entropy asks how much information the tally is expected to convey under a specified distribution of market compositions. The main results use the former because they do not assume a prior.


\section{Discussion and extensions}
\label{sec:discussion_extensions}

The minimal-tally result is tied to a particular institution. In the count-check mechanism, public information is used to compare individual reports with an independently certified feature of the market. For that purpose, the tally need only change when one participant changes her reported type. Other institutions may use prices, menus, audits, or individual evidence to discourage the same deviation. They may consequently require different public information.

Appendix~\ref{app:beyond_count_checks} studies this broader question. It asks which market censuses may share a public announcement when the mechanism is allowed to use transfers and allocations, rather than rejection alone, to provide incentives. In general, two neighboring censuses may sometimes be combined because prices or transfers make the relevant false report unattractive. The lower bound from the count-check mechanism does not apply to every possible institution. It extends to more general mechanisms under an additional condition. That is, neighboring censuses must be impossible to combine while preserving efficiency, individual rationality, budget balance, and ex post incentive compatibility. The appendix states this condition and gives economic circumstances under which it holds.

Certification raises a separate institutional question. The main analysis assumes that the public tally has an informational anchor, but it does not require any particular form of certification. Appendix~\ref{app:certification_extensions} considers certification by witnesses and by audits, and then allows the public certificate itself to be imperfect. These arrangements differ from self-certification because information or discipline enters from outside the unrestricted reports being checked.

The count check is also limited in its treatment of coordinated deviations. A stronger tally can detect a coalition when the coalition changes the reported census. Appendix~\ref{app:coalitions} studies how the required number of announcements varies with the largest coalition to be detected. No anonymous statistic of the census, however, can expose a permutation of reports that leaves every type count unchanged.

The appendix separates this informational limitation from the coalition's incentives. Under a symmetric outcome rule, a census-preserving permutation is profitable precisely when agents' preferences over trading positions form a cycle in which every participant weakly gains and at least one strictly gains. Strict self-selection rules out such cycles. These conditions strengthen the truthful equilibrium against coordinated deviations, but they do not provide full implementation. An unprofitable deviation from truth may still constitute another equilibrium once the corresponding false report profile has been reached.

This observation is closely related to the distinction between weak and full implementation. The count-check mechanism supports truthful reporting as an ex post equilibrium, but it may admit two other kinds of equilibrium. An accepted false profile is an equilibrium whenever it gives every agent a nonnegative true payoff, because any unilateral change causes rejection. A rejected profile may also be an equilibrium when no agent can restore consistency alone. Appendix~\ref{sec:weak_full_implementation} establishes both results and gives an accepted false profile that produces an inefficient allocation even when the entire census is public.

Ruling out profitable report permutations does not suffice for full implementation. Eliminating every inefficient equilibrium requires an additional instrument, such as identity-linked verification, audits, legal penalties, a procedure for revising inconsistent reports, or transfers that destabilize false report profiles.

The role of buyer information likewise depends on the institution. Buyer counts enter the tally because buyers report private valuation types in the direct mechanism. If buyer types are publicly known, that part of the tally is unnecessary. Buyer counts may also become unnecessary when posted prices separate buyer types without requiring reports. Appendix~\ref{app:buyer_side} develops these cases. Appendix~\ref{app:decentralized_prices} goes further and gives conditions under which posted prices guide buyers to an efficient allocation without a clearing rule. Those conditions require not only that buyers prefer the appropriate markets, but also that their choices respect capacity and that sellers do not benefit from entering markets intended for other qualities.

The finite nature of the argument is equally important. In a finite market, one false report changes an integer count and can change the tally. In a continuum economy, a single agent has measure zero and does not alter the aggregate distribution. Appendix~\ref{subsec:large_finite_continuum} develops this distinction and considers what survives along sequences of increasingly large finite markets. The number of announcements needed to detect a unilateral change remains bounded when the sets of qualities and buyer types are fixed, but the ability to detect that change disappears in the nonatomic limit.

Finally, the informational and institutional requirements can be viewed as substitutes. A direct mechanism can operate with limited public disclosure but takes an active role in collecting reports and assigning trades. A posted-price institution resembles decentralized exchange more closely, but generally requires certified prices, capacities, and a means of clearing applications. Appendix~\ref{subsec:privacy_objective} considers alternative measures of disclosure and explains why disclosure alone does not determine which institution is preferable. The amount of public information required for efficient trade depends on the work that the institution asks that information to perform.


\section{Conclusion}
\label{sec:conclusion}

Markets often use inefficient delay to learn what participants know privately. This paper considers an alternative procedure: certify a limited feature of market composition, and use it to discipline individual reports. In the count-check mechanism, the full census need not be disclosed. To reach efficient trade, it is enough for the public tally to change whenever one participant changes her reported type. Among tallies required to reject every unilateral false report, the least number of announcements is $\max\{K,L\}$ when every composition is possible, regardless of the number of buyers and sellers.

The result separates verification from disclosure. An institution may need detailed records to produce a trustworthy tally, but it need not reveal those records to market participants. The tally cannot, however, be reconstructed from the reports it is meant to check. It must rest on an independent source of information or discipline.

Limited disclosure also has limits. A tally may support truthful reporting without telling buyers how to divide themselves among scarce trading opportunities. Posted-price trade may require certified capacities and a clearing rule even when the incentive problem has been addressed. Information that verifies reports and information that organizes trade perform different tasks.

The informational needs of a market cannot be separated from its institutional design. A small public certificate may be enough when the institution collects reports and assigns trades. A more decentralized arrangement may require richer information or greater coordination. What must be disclosed depends on what the market asks the disclosure to accomplish.


\newpage
\appendix


\section{Proofs}
\label{app:proofs_main}


\subsection{Count-check mechanisms}
\begin{proof}[Proof of Theorem~\ref{thm:main_countcheck}]
 
First, suppose $\phi$ is separating. Fix a realized type profile $(q,\theta)$ with census $C_0=C(q,\theta)\in\Domain$, and certified tally $z_0=\phi(C_0)$. If every agent reports truthfully, the reported census is $\widehat C=C_0$, so the count check passes. The mechanism chooses $x^*(q,\theta)$ at date zero. Transfers are balanced pair by pair.

Every matched pair in $x^*(q,\theta)$ generates strictly positive surplus by construction of the selection rule. Hence, if $(i,j)\in x^*(q,\theta)$, then $v_{\theta_j}(q_i)-c(q_i)>0$. At the midpoint price, seller $i$ receives 
\[
\frac{v_{\theta_j}(q_i)+c(q_i)}{2}-c(q_i)=\frac{v_{\theta_j}(q_i)-c(q_i)}{2}>0,
\]
and buyer $j$ receives the same strictly positive surplus share. Unmatched agents receive zero.

Now consider a unilateral false report by a seller. If seller $i$ changes her report from $q^r$ to $q^s\neq q^r$, while all other agents report truthfully, the reported census either lies outside $\Domain$, in which case it is rejected, or is a one-agent neighbor of $C_0$. In the latter case, separation gives $\phi(\widehat C)\neq\phi(C_0)=z_0$, so the check fails. The deviating seller receives zero. Truthful reporting gives her a positive payoff if she is matched, and gives her zero otherwise. Thus, she cannot profitably deviate. The buyer argument is identical. It follows that truthful reporting is an ex post equilibrium, and the equilibrium outcome is immediate, efficient, and budget balanced.

Suppose now that a tally makes every unilateral false report fail the count check, but is not separating. Then, there exist one-agent neighboring censuses $C,C'\in\Domain$ with $\phi(C)=\phi(C')$. Because they are one-agent neighbors, there is a type profile with census $C$ and a unilateral false report that changes the reported census to $C'$. Under the certified tally $\phi(C)$, this false report passes the check because $\phi(C')=\phi(C)$. This contradicts the requirement that every unilateral false report fails. Hence, the tally must be separating.

It remains to prove the formula for the unrestricted domain. Let $R\coloneqq\max\{K,L\}$ and define $\phi$ by \eqref{eq:modular_tally}. If one seller changes her report from $q^r$ to $q^s$, with $r\neq s$, the tally changes by $s-r$ modulo $R$. Since $r,s\in\{0,\ldots,K-1\}$ and $R\geq  K$, the nonzero integer $s-r$ has absolute value strictly less than $R$ and is not congruent to zero modulo $R$. The same argument applies to a buyer who changes her report from $\theta^\ell$ to $\theta^h$. In this case, the tally changes by $h-\ell$, which is nonzero modulo $R$ because $R\geq  L$. Thus, $R$ labels are sufficient.

For necessity, fix any buyer census. On the seller side, take a base vector with $N-1$ sellers of quality $q^0$. For each $r\in\{0,\ldots,K-1\}$, add one seller of quality $q^r$. Any two of the resulting $K$ censuses differ only by the type of that one seller. A separating tally must assign all $K$ censuses different labels. Thus, at least $K$ labels are necessary. Similarly, fixing any seller census and fixing $M-1$ buyers at a base type gives $L$ buyer censuses that are pairwise one-agent neighbors. At least $L$ labels are necessary. It follows that at least $\max\{K,L\}$ labels are necessary, and the modular tally attains the bound.
   
\end{proof}


\begin{proof}[Proof of Corollary~\ref{cor:binary}]
A unilateral change in a seller's report changes $n_H$ by one, while a unilateral change in a buyer's report changes $m_H$ by one. Either change reverses the parity of $n_H+m_H$. The displayed tally separates every pair of one-agent neighbors and uses at most two announcements.

If $\Domain$ contains a pair of one-agent neighbors, those two censuses must receive different announcements, so at least two are necessary. If $\Domain$ contains no such pair, a constant tally is separating and one announcement is enough. The unrestricted domain contains one-agent neighbors. Hence, two announcements are necessary and sufficient on that domain.
\end{proof}


\begin{proof}[Proof of Proposition~\ref{prop:no_self_certification_main}]
Because $\phi$ is nonconstant, there are two profiles $\omega,\omega'\in\Omega_{\Domain}$ such that $\phi(C(\omega)) \neq \phi(C(\omega'))$. Truthful acceptance at $\omega'$ requires $A(m^T(\omega'))=1$.
The same message profile can be submitted when the true profile is $\omega$. Since the procedure observes only the messages, it must make the same decision and accept. At $\omega$, however, the message profile $m^T(\omega')$ reports a tally different from $\phi(C(\omega))$. Thus, it is tally-inconsistent. The procedure cannot reject it at $\omega$ while accepting the identical message at $\omega'$.
\end{proof}


\begin{proof}[Proof of Corollary~\ref{cor:self_generated_vacuous_main}]
The acceptance condition becomes $\phi(C(\hat\omega)) = \phi(C(\hat\omega))$, which holds for every $\hat\omega\in\Omega_{\Domain}$.
\end{proof}


\subsection{Posted-price clearing}
\label{app:proof_posted_main}

\begin{proof}[Proof of Lemma~\ref{lem:support}]
Fix a census $C$. Consider the linear relaxation of the type-class assignment
problem:
\[
        \max_{y\geq0} \sum_{\theta\in\Types} \sum_{q\in\Qualities} y_{\theta q}s_{\theta q}
\]
subject to
\[
        \sum_{\theta\in\Types}y_{\theta q}\leq n_q \quad\text{for every }q\in\Qualities
\]
and
\[
        \sum_{q\in\Qualities}y_{\theta q}\leq m_\theta \quad\text{for every }\theta\in\Types.
\]

The constraint matrix is the incidence matrix of a bipartite graph. Since the capacity vectors $n$ and $m$ have integer entries, the problem has an optimal solution with integer entries \citep[see, e.g.,][]{Kuhn1955,BurkardDellAmicoMartello2009}. Let $y^*$ be such a solution. It is feasible for the type-class assignment problem and maximizes total surplus.

The dual problem is
\[
        \min_{u,\rho\geq0} \left\{ \sum_{\theta\in\Types}m_\theta u_\theta + \sum_{q\in\Qualities}n_q\rho_q \right\}
\]
subject to
\[
        u_\theta+\rho_q\geq s_{\theta q} \quad \text{for every }(\theta,q)\in\Types\times\Qualities.
\]
Let $(u,\rho)$ be an optimal dual solution. Strong duality gives
\begin{equation}
        \sum_{\theta,q}y^*_{\theta q}s_{\theta q} = \sum_\theta m_\theta u_\theta + \sum_q n_q\rho_q.
        \label{eq:assignment_strong_duality}
\end{equation}

Complementary slackness gives
\[
        y^*_{\theta q}>0 \implies u_\theta+\rho_q=s_{\theta q}.
\]
It also gives
\[
        \sum_q y^*_{\theta q}<m_\theta \implies u_\theta=0
\]
and
\[
        \sum_\theta y^*_{\theta q}<n_q \implies \rho_q=0.
\]

Set $p_q=c(q)+\rho_q$. For every buyer type $\theta$ and quality $q$,
\[
v_\theta(q)-p_q= v_\theta(q)-c(q)-\rho_q = s_{\theta q}-\rho_q\leq u_\theta.
\]
If $y^*_{\theta q}>0$, complementary slackness makes the inequality an equality. Since $u_\theta\geq0$, an assigned type-$\theta$ buyer obtains her highest payoff among all qualities and the outside option. A seller who trades in market $q$ receives $p_q-c(q)=\rho_q\geq0$. This proves the result.
\end{proof}

\begin{proof}[Proof of Theorem~\ref{thm:posted}]
Fix a realized profile $(q,\theta)$ with certified census $C$, and let $(y^*,u,\rho,p)$ be supplied by Lemma~\ref{lem:support}. The organizer observes $C$ and publicly announces the resulting prices and seller capacities. Consider the strategy profile in which every seller of quality $q$ chooses market $q$ and every buyer uses the application rule stated in the main text.

Under the prescribed seller strategy, the number of sellers choosing market $q$ is $n_q$. Thus, the seller-capacity check passes.

For each buyer $j$, let
\[
        U_j \coloneqq \max\left\{ 0,\max_{q\in\Qualities} \left[v_{\theta_j}(q)-p_q\right] \right\}.
\]
We first show that $U_j=u_{\theta_j}$. If buyers of type $\theta_j$ are assigned under $y^*$, complementary slackness gives equality between $u_{\theta_j}$ and their payoff in every market to which they are assigned. The dual inequalities show that no other market gives them a greater payoff. If no buyer of that type is assigned, the corresponding buyer constraint is slack, so $u_{\theta_j}=0$. The dual inequalities show that every market gives a weakly negative payoff. In either case, $U_j=u_{\theta_j}$.

We next show that the clearing rule can find a matching satisfying its requirements. Since $y^*$ has integer entries, it can be implemented by an individual matching between buyers and sellers. Consider any such matching.

If $y^*_{\theta q}>0$, then $v_\theta(q)-p_q=u_\theta$. A buyer assigned to quality $q$ therefore applies to market $q$. Her application is strong when $u_\theta>0$ and standby when $u_\theta=0$. If $u_\theta>0$, complementary slackness implies
\[
        \sum_{q\in\Qualities}y^*_{\theta q}=m_\theta.
\]
All buyers of type $\theta$ are consequently matched under $y^*$. Thus, every buyer who submits a strong application can be assigned to a market to which she applies strongly.

Similarly, if $p_q>c(q)$, then $\rho_q>0$, and complementary slackness implies
\[
        \sum_{\theta\in\Types}y^*_{\theta q}=n_q.
\]
Every seller in market $q$ is matched. The individual matching that implements $y^*$ satisfies all the clearing requirements, so the clearing rule can proceed.

We now show that every matching selected by the clearing rule is efficient. For any feasible matching $x$,
\[
\begin{aligned}
        W(x\mid q,\theta)
        &=
        \sum_{(i,j)\in x}
        \left[v_{\theta_j}(q_i)-p_{q_i}\right]
        +
        \sum_{(i,j)\in x}
        \left[p_{q_i}-c(q_i)\right] \\
        &\leq
        \sum_{j\in\Buyers}U_j
        +
        \sum_{i\in\Sellers}\rho_{q_i}.
\end{aligned}
\]
The inequality follows because no buyer can obtain more than $U_j$ at the posted prices and every $\rho_q$ is nonnegative. Since $U_j=u_{\theta_j}$, the right-hand side is
\[
        \sum_{\theta\in\Types}m_\theta u_\theta + \sum_{q\in\Qualities}n_q\rho_q.
\]
By \eqref{eq:assignment_strong_duality}, this is the first-best surplus.

Under the intended application strategy, every matched buyer $j$ receives $U_j$, and every buyer with $U_j>0$ is matched. Likewise, every seller with $\rho_q>0$ is matched. Buyers and sellers with zero supported payoffs contribute nothing whether or not they trade. The matching selected by the clearing rule attains the first-best surplus, and every participant receives a nonnegative payoff.

The minimum-trade convention also rules out zero-surplus trades. Such a trade would give both the buyer and the seller zero. The buyer's application would be standby, and the seller's market would satisfy $p_q=c(q)$. Removing the pair would leave every strong applicant assigned and every seller with $p_q>c(q)$ matched. The remaining matching would still satisfy the clearing requirements but would contain one fewer trade. This contradicts the minimum-trade convention. Hence, the selected matching contains no zero-surplus trade.

It remains to verify incentives. A seller who follows the prescribed strategy receives a nonnegative payoff. If she alone chooses a market bearing a different quality label, the number of sellers in her original market falls by one and the number in the new market rises by one. The capacity check fails, trade is cancelled, and she receives zero. Such a deviation is not profitable.

Under the prescribed strategy, buyer $j$ receives $U_j$. Consider any deviation in her applications or in their designations. If the deviation makes clearing impossible, trade is cancelled and she receives zero. If clearing remains possible, she is either unmatched or assigned to some market $q$. Her payoff is then at most
\[
        \max\{0,v_{\theta_j}(q)-p_q\}\leq U_j.
\]
No buyer has a profitable deviation. Since these arguments hold at every realized profile, the prescribed strategies form an ex post equilibrium.

Finally, fix any common prior over type profiles and suppose that the organizer's rule for selecting prices and capacities from each census is common knowledge. After observing the public price and capacity announcement, participants update their beliefs by Bayes' rule wherever possible. On the equilibrium path, seller choices reveal the quality labels prescribed by the strategy. The prescribed buyer applications are optimal at every profile consistent with the announcement.

Consider an off-path seller-choice history at which the capacity check passes. The number of sellers in every market then agrees with the announced seller capacities. Whenever Bayes' rule does not determine beliefs, buyers may believe that every seller choosing market $q$ has quality $q$. Given these beliefs, the prescribed buyer applications remain sequentially optimal. At histories where the capacity check fails, trade is cancelled and beliefs do not affect payoffs. The strategies and these beliefs form a pure-strategy perfect Bayesian equilibrium.
\end{proof}


\begin{proof}[Proof of Proposition~\ref{prop:seller_capacity_lower_bound}]
Fix $C\in\Domain$. If sellers follow the intended strategy, precisely the sellers of quality $q$ choose market $q$. The number of sellers choosing each market is thus $n(C)$. Since the announced capacity depends only on the public signal, implementation at $C$ requires $\overline n(\psi(C))=n(C)$.

Now suppose that $\psi(C)=\psi(C')$. The organizer must announce the same capacity vector at the two censuses. Implementation requires this vector to equal both $n(C)$ and $n(C')$, so $n(C)=n(C')$. Censuses with different seller count vectors must consequently produce different signals. On the unrestricted domain, there are
\[
        \binom{N+K-1}{K-1}
\]
nonnegative integer seller count vectors whose entries sum to $N$. Each must produce a different public signal. This completes the proof.
\end{proof}


\subsection{Incentives, coordination, and privacy}

\begin{proof}[Proof of Proposition~\ref{prop:incentive_coordination_main}]
There are two sellers, $A$ and $B$, each with a good of the same publicly known quality $q$ and cost $c(q)=0$. Both goods are offered at the same price $p\in(0,1)$. There are three buyers, indexed by $1,2,3$, and two buyer types, $\theta_L$ and $\theta_H$, with $v_{\theta_L}(q)=0$ and $v_{\theta_H}(q)=1$. Consider the census with two high-valuation buyers and one low-valuation buyer. Efficiency requires the two high-valuation buyers to purchase the two goods, one from each seller, while the low-valuation buyer takes the outside option.

There is only one seller quality and there are two buyer types. By Theorem~\ref{thm:main_countcheck}, a one-bit tally is sufficient for the count-check mechanism. In particular,
\[
        \phi(C)=m_{\theta_H}\bmod 2
\]
changes whenever one buyer changes her reported type.

Now disclose the full census. Every buyer knows that two of the three buyers have high valuations. A high-type buyer knows that one of the other two buyers is also a high type, but she does not know which one. A low-type buyer can instead infer that both other buyers are high types. The argument below concerns the choices of high-type buyers, whose uncertainty remains. 

Consider first deterministic behavior. Let $f(k)\in\{A,B\}$ be the seller chosen by buyer $k$ when she is a high type. The choice may depend on her identity, her own type, and the public census. For immediate efficiency at every assignment of the two high types, any two distinct buyer identities $k$ and $\ell$ must satisfy $f(k)\neq f(\ell)$. Otherwise, when $k$ and $\ell$ are the two high types, they choose the same seller and leave the other seller unused. But three buyer identities cannot be assigned pairwise different choices from the two-element set $\{A,B\}$. As a result, some assignment of high types produces congestion.

Randomization does not remove the problem. Let $X_k\in\{A,B\}$ be the possibly random seller chosen by buyer $k$ when she is a high type. The variables may depend on public or private randomization, but there is no recommendation conditioned on the realized identities of the two high types. If every possible pair of high-type buyers were separated with probability one, then $X_i\neq X_j$, for $i\neq j$ and $i,j\in\{1,2,3\}$, would all hold with probability one. This would require three pairwise different values in the set $\{A,B\}$, which is impossible. For at least one assignment of the two high types, congestion occurs with positive probability.
\end{proof}


\begin{proof}[Proof of Lemma~\ref{prop:disclosure_gap_main}]
A seller census is a vector of $K$ nonnegative integers whose entries sum to $N$. The number of such vectors is
\[
        \binom{N+K-1}{K-1}.
\]
Similarly, the number of buyer count vectors over $L$ types is
\[
        \binom{M+L-1}{L-1}.
\]
The two sides can vary independently on the unrestricted domain, so the number of possible market censuses is the product in \eqref{eq:full_census_messages}. Full census disclosure must distinguish all of them. Theorem~\ref{thm:main_countcheck} shows that a separating tally requires only $\max\{K,L\}$ announcements.
\end{proof}



\section{Extensions}

This appendix develops extensions of the main model. It considers alternative sources of certification, public information outside the count-check mechanism, coordinated deviations, and the distinction between finite and nonatomic markets. It also examines the roles of buyer information, full implementation, decentralized posted prices, and imperfect certification. None of these extensions is needed for the results in the main text. Their purpose is to clarify which conclusions depend on the count-check institution and which extend to other informational and trading arrangements.

\subsection{Certification extensions}
\label{app:certification_extensions}

Section~\ref{sec:certification} shows that a tally cannot acquire evidentiary force from the reports it is meant to check. This subsection considers two well-known sources of outside discipline, i.e., witnesses who observe the tally, and audits that verify individual reports. Neither arrangement is intended as a general theory of certification. Their purpose is to illustrate how an informational anchor may enter the mechanism.

\paragraph{Witnesses.} Suppose that two or more witnesses observe the true tally, $ z^0=\phi(C^0)$. Each witness reports one tally value. If all reports agree, the common value is announced and the count-check mechanism is run. If the reports disagree, no trade occurs. A witness receives a reward $r>0$ when unanimity is reached and nothing otherwise.

The reward is treated as a certification cost paid by the organizer. It is separate from the transfers between buyers and sellers.

\begin{lemma}
\label{lem:witness_certification}
Suppose that $\phi$ is separating on $\Domain$. The witness mechanism has an ex post equilibrium in which every witness reports $z^0$, market participants report their types truthfully, and the efficient matching is carried out at date zero.
\end{lemma}

\begin{proof}
Consider the strategy profile in which every witness reports $z^0$ and, after its announcement, every market participant reports truthfully. The witness reports agree, and each witness receives $r$. Because $\phi$ is separating, Theorem~\ref{thm:main_countcheck} supports truthful reporting by market participants.

Fix one witness and suppose that all other witnesses report $z^0$. Reporting $z^0$ preserves unanimity and gives the witness the reward $r$, together with any nonnegative trading payoff she receives if she is also a market participant. Reporting another value breaks unanimity. The mechanism then stops, and the witness receives neither the reward nor a trading payoff. The deviation is not profitable. Deviations by market participants are unprofitable by Theorem~\ref{thm:main_countcheck}.
\end{proof}

Unanimity supports truthful certification, but it does not select truth on its own. The witnesses may also coordinate on a common false announcement. The reward encourages agreement, not accuracy. Credible witness certification requires some further reason for witnesses to report what they observe, such as legal liability, reputation, occasional verification, conflicting interests, or institutional separation. The lemma establishes the existence of a truthful equilibrium but not its uniqueness.

\paragraph{Audits.} A second possibility is to verify individual reports directly. This is a stronger intervention than aggregate certification, but it provides a simple way to obtain reliable information about the census.

\begin{assumption}
\label{ass:audit}
After agents report their types, every false report is detected with probability at least $\alpha\in(0,1]$. A detected agent is excluded from trade and pays a penalty $F\geq0$. If the false report is not detected, its gain relative to truthful reporting is at most $G\geq0$.
\end{assumption}

The bound $G$ exists whenever the set of types and reports is finite and the mechanism places finite bounds on transfers. The following condition will also be used in later audit extensions,
\begin{equation}
        \alpha F\geq(1-\alpha)G.
        \label{eq:audit_dominance}
\end{equation}

\begin{lemma}
\label{lem:audit_truthfulness}
Under Assumption~\ref{ass:audit} and \eqref{eq:audit_dominance}, truthful reporting is an ex post equilibrium of the audited reporting stage.
\end{lemma}

\begin{proof}
Fix an agent, a realized type profile, and truthful reports by all other agents. Let $u^T\geq0$ be the agent's payoff from reporting truthfully. Consider a false report, and let $\beta\geq\alpha$ be its probability of detection.

If the false report is not detected, its payoff gain relative to truth is at most $G$. If it is detected, the agent receives $-F$, so its gain relative to truth is $-F-u^T\leq-F$. The expected gain from the false report is at most
\[
        (1-\beta)G-\beta F \leq (1-\alpha)G-\alpha F \leq0,
\]
where the last inequality follows from \eqref{eq:audit_dominance}. No false report is profitable.
\end{proof}

In the truthful equilibrium, the reports reveal the census and allow the organizer to calculate any desired tally. The source of discipline is the audit rather than the tally. This arrangement should be understood as one way to produce reliable aggregate information, not as an alternative proof that a self-generated tally can discipline unrestricted reports.

The two results, while already well-known, expose different limitations. Witness certification requires an equilibrium selection or truth-telling institution in addition to unanimity. Audit certification requires a verification technology and penalties large enough to discourage false reports. Neither source is costless, but both supply information or discipline that is absent from the reports alone.


\paragraph{Noisy certificates.} The main results assume that the public tally is certified correctly. Suppose instead that certification is imperfect. At the true census $C$, let the announced certificate be a public random variable $\widetilde z\in\mathcal Z$ satisfying
\begin{equation}
        \Pr\{ \widetilde z=\phi(C) \mid C \} \geq1-\varepsilon.
        \label{eq:noisy_certificate}
\end{equation}
No restriction is placed on the announcement made with the remaining probability.

Consider the following audited count-check mechanism. Agents report their types, producing a reported census $\widehat C$. If $\phi(\widehat C)\neq\widetilde z$, trade is cancelled and no transfers are made. If $\phi(\widehat C)=\widetilde z$, the mechanism applies the usual count-check allocation and transfers. Any false report that passes the tally comparison is subject to the audit technology in Assumption~\ref{ass:audit}.

Let $G$ be the bound in that assumption, taken uniformly over true profiles, reported profiles, and realizations of $\widetilde z$. Retain the audit condition $\alpha F\geq(1-\alpha)G$ from \eqref{eq:audit_dominance}.

\begin{proposition}
\label{prop:imperfect_certification}
Under Assumption~\ref{ass:audit} and \eqref{eq:audit_dominance}, truthful reporting is an ex post equilibrium after every realization of the public certificate. In the truthful equilibrium, the efficient allocation is carried out at date zero with probability at least $1-\varepsilon$.
\end{proposition}

\begin{proof}
Fix a realized type profile with census $C$ and a realized certificate $\widetilde z$. If agents report truthfully and $\widetilde z=\phi(C)$, the tally comparison passes and the count-check mechanism carries out the efficient allocation. If $\widetilde z\neq\phi(C)$, truthful reports fail the comparison and every agent receives zero. Truthful reporting gives every agent a nonnegative payoff for every certificate realization.

Now consider a false report by one agent while all other agents report truthfully. If the false report fails the tally comparison, the agent receives zero and does not improve on truthful reporting. If it passes, the report is subject to the audit. Lemma~\ref{lem:audit_truthfulness} and condition~\eqref{eq:audit_dominance} imply that its expected gain relative to truthful reporting is nonpositive. No false report is profitable.

Thus, truthful reporting is an ex post equilibrium for every realized certificate. By \eqref{eq:noisy_certificate}, the certificate agrees with the true tally with probability at least $1-\varepsilon$. On that event, truthful reports pass the comparison and the efficient allocation is carried out at date zero.
\end{proof}

Certification accuracy and reporting incentives play different roles. The error probability $\varepsilon$ determines how often truthful reports are accepted. The audit condition determines whether an agent wishes to exploit an incorrect certificate. This is why \eqref{eq:audit_dominance} does not depend on $\varepsilon$, as it controls the payoff from a false report conditional on that report passing the tally comparison.

Without audits or another source of discipline, an incorrect certificate may create an incentive to misreport. If it coincides with the tally of a neighboring census, a false report may pass when truthful reports fail. An agent whose truthful payoff is zero may then benefit from changing her report. The audit condition removes this gain but does not restore reliability. Efficiency still depends on the certificate being correct.


\subsection{Minimal public information beyond count checks}
\label{app:beyond_count_checks}

Theorem~\ref{thm:main_countcheck} characterizes the information required by the count-check mechanism. A broader question is whether another institution could implement efficient trade with less public information by using prices or transfers to discourage false reports.

Throughout this subsection, the public signal is a certified function $\psi:\Domain\to\mathcal Z$ of the true market census. For each announcement $z\in\mathcal Z$, define its \emph{signal cell} by
\[
        \Domain_z = \{C\in\Domain\mid \psi(C)=z\}.
\]
Two censuses are \emph{pooled} when they belong to the same signal cell and produce the same public announcement.

After observing $z$, the institution may ask agents to report their types and may condition allocations and transfers on those reports. It has no independent evidence of individual types. Reports that imply a census outside $\Domain_z$ can be rejected because they are inconsistent with the public signal. Reports that remain within $\Domain_z$ must instead be disciplined by the allocation and transfer rule.

Write $\Omega \coloneqq \Qualities^\Sellers\times\Types^\Buyers$ for the set of type profiles, with a typical profile denoted by $\omega=(q,\theta)$. For any set of censuses $\mathcal E\subseteq\Domain$, let
\[
        \Omega(\mathcal E) \coloneqq \{\omega\in\Omega \mid C(\omega)\in\mathcal E\}.
\]

An outcome is a pair $o=(x,\tau)$, where $x\in\Matchings$ is a feasible matching and
\[
        \tau \coloneqq (\tau_h)_{h\in\Sellers\cup\Buyers} \in\mathbb R^{\Sellers\cup\Buyers}
\]
is a vector of transfers. A positive transfer is received by the agent. At profile $\omega=(q,\theta)$, seller $i$ receives
\[
        u_i(o\mid\omega) = \tau_i - \mathds{1}\{i\text{ is matched in }x\}c(q_i),
\]
and buyer $j$ receives
\[
        u_j(o\mid\omega) = \tau_j + \sum_{\{i \,\mid\, (i,j)\in x\}}v_{\theta_j}(q_i).
\]
The outcome is budget balanced if
\[
        \sum_{h\in\Sellers\cup\Buyers}\tau_h=0,
\]
and it is ex post individually rational at $\omega$ if $u_h(o\mid\omega)\geq0$ for every $h\in\Sellers\cup\Buyers$.

Let $\mathcal O^*(\omega)$ be the set of outcomes that maximize total surplus at $\omega$, balance the budget, and satisfy ex post individual rationality. This set is nonempty, as the selected efficient matching together with bilateral midpoint transfers belongs to it.

\begin{definition}
\label{def:implementable_block}
A nonempty set of censuses $\mathcal E\subseteq\Domain$ is an \textbf{implementable census block} if there is a direct outcome rule
\[
        f_{\mathcal E}:\Omega(\mathcal E) \to  \Matchings\times \mathbb R^{\Sellers\cup\Buyers}
\]
with the following properties,

\begin{enumerate}[label=\textup{(\roman*)},leftmargin=2em]
    \item $f_{\mathcal E}(\omega)\in\mathcal O^*(\omega)$ for every $\omega\in\Omega(\mathcal E)$;

    \item for every agent $h$, every $\omega\in\Omega(\mathcal E)$, and every alternative report $\hat\omega_h$ satisfying $C(\hat\omega_h,\omega_{-h})\in\mathcal E$, one has
    \begin{equation}
            u_h(f_{\mathcal E}(\omega)\mid\omega) \geq u_h( f_{\mathcal E}(\hat\omega_h,\omega_{-h}) \mid\omega).
            \label{eq:IC_E}
    \end{equation}
\end{enumerate}
\end{definition}

The first condition requires efficient, balanced, and individually rational outcomes throughout the block. The second requires truthful reporting against every unilateral deviation that remains compatible with the same public announcement. A deviation that leaves the block can be rejected by comparing its reported census with the public signal.

The next result identifies the public information required within this class of direct mechanisms.

\begin{lemma}
\label{lem:block_characterization}
A public signal $\psi:\Domain\to\mathcal Z$ supports immediate efficient trade with budget balance and ex post individual rationality through a direct no-verification mechanism if and only if every signal cell $\Domain_z$ is an implementable census block.

As a result, the least number of public announcements in this class is
\begin{equation}
        \kappa(\Domain) = \min_{\mathcal P}|\mathcal P|,
        \label{eq:kappa_partition}
\end{equation}
where the minimum is taken over all partitions $\mathcal P$ of $\Domain$ into implementable census blocks.
\end{lemma}

\begin{proof}
Suppose first that $\psi$ supports the stated implementation. Fix an announcement $z$, and let $g_z(\omega)$ be the truthful equilibrium outcome at each $\omega\in\Omega(\Domain_z)$. Implementation requires $g_z(\omega)\in\mathcal O^*(\omega)$.

Consider an agent $h$, a profile $\omega\in\Omega(\Domain_z)$, and an alternative report $\hat\omega_h$ whose implied census also belongs to $\Domain_z$. Since the public announcement is unchanged, ex post incentive compatibility requires $u_h(g_z(\omega)\mid\omega) \geq u_h(g_z(\hat\omega_h,\omega_{-h})\mid\omega)$. Thus, the restriction of $g_z$ to $\Omega(\Domain_z)$ satisfies Definition~\ref{def:implementable_block}, and $\Domain_z$ is implementable.

Conversely, suppose every signal cell is implementable, and let $f_z$ be an implementing rule for $\Domain_z$. After observing $z$, agents report their types. If the reported census lies outside $\Domain_z$, the mechanism selects no trade and makes no transfers. If it lies in $\Domain_z$, the mechanism applies $f_z$.

Truthful reporting produces an outcome in $\mathcal O^*(\omega)$. A false report that remains within $\Domain_z$ is unprofitable by \eqref{eq:IC_E}. A false report that leaves $\Domain_z$ produces zero payoff, which cannot improve on the nonnegative truthful payoff. Truthful reporting is an ex post equilibrium, and its outcome is immediate, efficient, budget balanced, and ex post individually rational.

Finally, a public signal divides $\Domain$ into its signal cells. Conversely, any partition of $\Domain$ into implementable blocks can be used as a public signal. Minimizing the number of cells gives \eqref{eq:kappa_partition}.
\end{proof}

Lemma~\ref{lem:block_characterization} also clarifies the scope of Theorem~\ref{thm:main_countcheck}. The count-check mechanism rejects every unilateral report that moves the census away from the public announcement. A general mechanism does not need to reject every such report because prices or transfers may make some of them unprofitable even when the corresponding censuses receive the same announcement.

Recall from Section~\ref{sec:certified_count_checks} that $C\sim_1C'$ denotes two one-agent neighboring censuses.

\begin{definition}
\label{def:local_nonpoolability}
A pair of one-agent neighbors $C,C'\in\Domain$ is \textbf{locally nonpoolable} if $\{C,C'\}$ is not an implementable census block. The domain satisfies \textbf{local nonpoolability} if every pair of one-agent neighbors in $\Domain$ is locally nonpoolable.
\end{definition}

Local nonpoolability is a condition on the economic environment, not on the count-check mechanism. It says that if two neighboring censuses receive the same public announcement, no allocation and transfer rule can implement efficient trade at both while preserving budget balance, ex post individual rationality, and ex post incentive compatibility. Under this condition, the count-check requirement also becomes necessary for the broader class of direct no-verification mechanisms.

\begin{proposition}
\label{prop:neighbor_separation_general}
Suppose that $\Domain$ satisfies local nonpoolability. A public signal $\psi:\Domain\to\mathcal Z$ supports immediate efficient trade with budget balance and ex post individual rationality through a direct no-verification mechanism if and only if it separates one-agent neighbors, i.e.,
\begin{equation*}
        C\sim_1C' \implies \psi(C)\neq\psi(C') \qquad \text{for every }C,C'\in\Domain.
\end{equation*}
On the unrestricted domain, the least number of public announcements is $\max\{K,L\}$.
\end{proposition}

\begin{proof}
For necessity, suppose that $\psi(C)=\psi(C')$ for some $C\sim_1C'$. The two censuses belong to the same signal cell. By Lemma~\ref{lem:block_characterization}, that cell must be implementable. Any subset of an implementable block is implementable by restricting its outcome rule, so $\{C,C'\}$ would also be implementable. This contradicts local nonpoolability.

For sufficiency, suppose that $\psi$ separates one-agent neighbors. Use $\psi$ as the certified tally in the count-check mechanism. A unilateral false report either produces a census outside $\Domain$ or moves the census to a one-agent neighbor with a different public announcement. In either case the report is rejected. Theorem~\ref{thm:main_countcheck} gives a truthful ex post equilibrium with immediate efficient trade, budget balance, and ex post individual rationality.

On the unrestricted domain, equation \eqref{eq:modular_tally} supplies a separating tally with $\max\{K,L\}$ announcements. The seller-side and buyer-side arguments in Theorem~\ref{thm:main_countcheck} show that fewer announcements cannot separate every neighboring pair.
\end{proof}

The local nonpoolability condition can be checked by comparing the payoff from imitating a neighboring type with the greatest payoff available under truthful efficient implementation. For an agent $h$, define
\[
        \overline U_h^*(\omega) = \sup_{o\in\mathcal O^*(\omega)} u_h(o\mid\omega).
\]
Since the sum of agents' utilities in any balanced efficient outcome is $W^*(\omega)$ and all utilities are nonnegative, we have that
\[
        \overline U_h^*(\omega)\leq W^*(\omega)<\infty.
\]

\begin{lemma}
\label{lem:spoofing_test}
Let $\omega,\omega'\in\Omega(\Domain)$ differ only in the type of agent $h$. If
\begin{equation}
        \inf_{o\in\mathcal O^*(\omega)} u_h(o\mid\omega') > \overline U_h^*(\omega'),
        \label{eq:spoofing_test}
\end{equation}
then the censuses $C(\omega)$ and $C(\omega')$ are locally nonpoolable.
\end{lemma}

\begin{proof}
Suppose instead that $\{C(\omega),C(\omega')\}$ is implementable, and let $f$ be an implementing rule. Truthful implementation at $\omega'$ gives agent $h$ no more than $\overline U_h^*(\omega')$. If that agent instead reports the type she has at $\omega$, she obtains $f(\omega)$ evaluated at her true type in $\omega'$. Her payoff from doing so is at least
\[
        \inf_{o\in\mathcal O^*(\omega)} u_h(o\mid\omega'),
\]
which is strictly greater by \eqref{eq:spoofing_test}. This violates ex post incentive compatibility.
\end{proof}

For sellers, the condition has a familiar adverse-selection interpretation. Let $\omega^H$ and $\omega^L$ differ only in seller $i$'s quality, with $c(q_H)>c(q_L)$. Suppose seller $i$ is matched in every efficient outcome at $\omega^H$. Ex post individual rationality then requires her transfer at $\omega^H$ to be at least $c(q_H)$. If the low-cost seller at $\omega^L$ imitates the high-cost seller, her payoff is at least $c(q_H)-c(q_L)$. It follows from Lemma~\ref{lem:spoofing_test} that the two censuses are locally nonpoolable whenever
\begin{equation}
        c(q_H)-c(q_L) > \overline U_i^*(\omega^L).
        \label{eq:seller_nonpoolability}
\end{equation}

\begin{corollary}
\label{cor:seller_neighbor_lower_bound}
Suppose that, for every pair of seller-side one-agent neighbors in $\Domain$, there are profiles $\omega^H,\omega^L$ with those censuses that satisfy \eqref{eq:seller_nonpoolability}. Then, any public signal supporting immediate efficient trade through a direct no-verification mechanism must separate every seller-side one-agent neighbor. If the seller census domain is unrestricted, at least $K$ public announcements are necessary.
\end{corollary}

\begin{proof}
Equation~\eqref{eq:seller_nonpoolability} and Lemma~\ref{lem:spoofing_test} imply that every seller-side neighboring pair is locally nonpoolable. Such a pair cannot belong to one implementable signal cell by Lemma~\ref{lem:block_characterization}.

For the lower bound, fix all buyers and all but one seller. Allow the remaining seller to have any of the $K$ qualities. The resulting $K$ censuses are pairwise seller-side one-agent neighbors and must receive different announcements.
\end{proof}

Without local nonpoolability, prices or menus may screen types even when neighboring censuses receive the same public announcement.

\begin{proposition}
\label{prop:no_unconditional_lower_bound}
There are environments in which immediate efficient trade is implementable with one public announcement even though the census domain contains one-agent neighbors.
\end{proposition}

\begin{proof}
Let $N=M=1$. There is one publicly known seller quality $q$ with $c(q)=1$, and the buyer has type $\theta_L$ or $\theta_H$, with $v_{\theta_L}(q)=0$ and $v_{\theta_H}(q)=2$. The two buyer censuses are one-agent neighbors. Use a constant public signal and offer trade at price $p=3/2$. Equivalently, in a direct mechanism, trade at this price if the buyer reports $\theta_H$ and do not trade if she reports $\theta_L$. The high type reports truthfully and receives $1/2$ from trade. The low type also reports truthfully, since falsely reporting $\theta_H$ would give her payoff $-3/2$. The seller receives $1/2$ whenever trade occurs.

The mechanism is efficient, budget balanced, and ex post individually rational, and truthful reporting is an ex post equilibrium. Thus, the two neighboring buyer censuses can share one public announcement because the posted price screens the buyer types.
\end{proof}

The general public-information problem is a problem of finding which censuses can share an announcement while remaining jointly implementable. The partition formula in Lemma~\ref{lem:block_characterization} gives the answer for direct no-verification mechanisms. Separation of one-agent neighbors is always sufficient through the count check, but it is necessary for the broader class only when neighboring censuses cannot be combined by other incentive instruments.


\subsection{Coalitional detectability}
\label{app:coalitions}

The main theorem considers a false report by one agent. This extension asks what a tally can detect when several agents coordinate their reports. A \emph{coalition} is a set of agents who choose their reports jointly.

As in the unilateral case, detectability is a stronger requirement than equilibrium alone. The analysis below asks the tally to expose every census-changing joint false report by a coalition of the specified size. Some of these reports may be unprofitable. A strong-equilibrium criterion would consider only deviations that make every coalition member weakly better off and at least one member strictly better off, while coalition-proofness would add the requirement that the deviation be self-enforcing. The results below provide sufficient protection against census-changing coalitional deviations, but they do not characterize the least public information required for strong or coalition-proof implementation.

Let $C=(n,m)$ and $C'=(n',m')$ be two censuses with the same total numbers of sellers and buyers. Define the seller-side and buyer-side distances
\[
        d_S(C,C') = \frac12 \sum_{q\in\Qualities}|n_q-n'_q|
\]
and
\[
        d_B(C,C') = \frac12 \sum_{\theta\in\Types}|m_\theta-m'_\theta|.
\]
Their sum, $d(C,C')=d_S(C,C')+d_B(C,C')$, is the smallest number of individual type reports that must be changed to transform census $C$ into census $C'$. In particular, we obtain $d(C,C')=1$ if and only if $C$ and $C'$ are one-agent neighbors.

\begin{definition}
\label{def:r_separating}
Fix an integer $r\geq1$. A tally $\phi:\Domain\to\mathcal Z$ is \textbf{$r$-separating} if
\[
        1\leq d(C,C')\leq r \implies \phi(C)\neq\phi(C')
\]
for every $C,C'\in\Domain$.
\end{definition}
The separating tallies of Section~\ref{sec:certified_count_checks} are the special case where $r=1$.

\begin{lemma}
\label{lem:coalitional_count_checks}
A tally rejects every census-changing false-report profile submitted by a coalition of at most $r$ agents if and only if it is $r$-separating. 

When an $r$-separating tally is used in the count-check mechanism, no census-changing deviation by a coalition of at most $r$ agents can make every coalition member weakly better off and at least one member strictly better off.
\end{lemma}

\begin{proof}
Suppose first that $\phi$ is $r$-separating. Let the true census be $C$, and consider a coalition $G$ with $|G|\leq r$. If only members of $G$ change their reports, the resulting census $C'$ satisfies $d(C,C')\leq|G|\leq r$. If the deviation changes the census, then $C'\neq C$, and $r$-separation gives $\phi(C')\neq\phi(C)$. The count check fails, so trade is cancelled and all transfers are zero.

Conversely, suppose that $\phi$ is not $r$-separating. There are distinct censuses $C,C'\in\Domain$ such that $d(C,C')\leq r$ and $\phi(C)=\phi(C')$. By the definition of $d(C,C')$, a coalition of at most $r$ agents can change its reports at a profile with census $C$ so that the reported census becomes $C'$. The tally is unchanged, and the false reports pass the count check.

For the final statement, a detected deviation gives every coalition member zero. Under truthful reporting, every agent receives a nonnegative payoff. The deviation cannot make all members weakly better off and one strictly better off. This completes the proof.
\end{proof}

A simple construction gives protection against any fixed coalition size. Choose an arbitrary base seller quality $q^0$ and an arbitrary base buyer type $\theta^0$. For every other category, announce its count after taking the remainder upon division by $r+1$.

\begin{proposition}
\label{prop:residue_tallies}
For every $r\geq1$, the tally
\begin{equation*}
        \phi_r(C) = \left( (n_q\bmod(r+1))_{q\in\Qualities\setminus\{q^0\}}, (m_\theta\bmod(r+1))_{\theta\in\Types\setminus\{\theta^0\}} \right)
\end{equation*}
is $r$-separating. It uses at most $(r+1)^{K+L-2}$ public announcements.
\end{proposition}

\begin{proof}
Suppose that $\phi_r(C)=\phi_r(C')$. For every nonbase seller quality $q\neq q^0$, the difference $n_q-n'_q$ is a multiple of $r+1$. Because the seller counts sum to $N$, the difference in the base coordinate, $n_{q^0}-n'_{q^0}$, is also a multiple of $r+1$.

If the seller count vectors differ, at least one coordinate rises and another falls. Their $\ell_1$-distance is at least $2(r+1)$, which implies $d_S(C,C')\geq r+1$. The same argument applies to the buyer counts. Hence,
\[
        C\neq C' \ \text{ and }\ \phi_r(C)=\phi_r(C') \implies d(C,C')\geq r+1.
\]
No two censuses at distance at most $r$ receive the same announcement, so $\phi_r$ is $r$-separating.

There are $K+L-2$ announced coordinates, and each has at most $r+1$ possible remainders. Therefore, the number of possible announcements is at most $(r+1)^{K+L-2}$.
\end{proof}

The construction is not claimed to minimize the number of announcements when $r>1$. Its significance is that, for fixed $r$, $K$, and $L$, the number does not grow with the numbers of buyers and sellers. Moreover, even a tally that minimizes the number of announcements among $r$-separating tallies need not minimize the information required to deter profitable coalitional deviations. The latter form a subset of the deviations covered by $r$-separation and depend on the allocation and transfer rule.

The preceding results concern coalitions whose reports change the census. Aggregate information faces a more fundamental limitation when a coalition preserves all type counts.

\begin{lemma}
\label{prop:census_preserving_swaps}
Suppose that, at some feasible profile, two agents on the same side of the market have different types. Those agents can exchange their reports without changing the market census. No tally that depends only on the census can detect this deviation, even if the tally reveals the full census.
\end{lemma}

\begin{proof}
Consider two sellers $i$ and $k$ with qualities $q_i\neq q_k$. The buyer case is the same. Let seller $i$ report $q_k$ and seller $k$ report $q_i$. One reported seller is removed from each of the two true quality categories and one is added back to each. Every seller count is unchanged. Since the reported census is the true census, every function of that census takes the same value before and after the exchange of reports. A census-based check that accepts truthful reports must also accept the exchanged reports.
\end{proof}

Lemma~\ref{prop:census_preserving_swaps} is a statement about detectability, not profitability. The exchange of reports may be unattractive under a particular allocation or transfer rule. But if it is profitable, or if it supports another equilibrium, no anonymous statistic of market composition can expose it. Deterring such deviations requires identity-linked evidence, individual audits, or incentives supplied by the trading rule itself.

\paragraph{Profitable report permutations.} The previous lemma shows that a permutation of reports on one side of the market may leave every count unchanged. Whether such a permutation is attractive depends on the allocation and transfers that the reports produce. This paragraph holds the count-check outcome rule fixed and characterizes profitable permutations.

Use the notation of Appendix~\ref{app:beyond_count_checks}, and let $f^{\mathrm{cc}}(\widehat\omega)$ denote the allocation and transfers selected by the count-check mechanism when the report profile $\widehat\omega$ is accepted. Fix a true profile $\omega$ and one side of the market, $A\in\{\Sellers,\Buyers\}$. For an agent $h\in A$, write $t_h$ for her type, $t_h=q_h$ on the seller side and $t_h=\theta_h$ on the buyer side.

Let $G\subseteq A$ be a coalition, and let $\sigma:G\to G$ be a permutation. The associated report profile $\widehat\omega^\sigma$ is defined by
\[
    \widehat t_h^\sigma=t_{\sigma(h)} \quad\text{for }h\in G, 
\]
\[
 \widehat t_h^\sigma=t_h \quad\text{for }h\notin G.
\]
Without loss, $G$ contains only agents whose reports change, so that $t_{\sigma(h)}\neq t_h$ for every $h\in G$. Because $\sigma$ merely rearranges the reports within one side of the market, we have that $C(\widehat\omega^\sigma)=C(\omega)$. The permutation passes every census-based count check that accepts truthful reports.

\begin{definition}
\label{def:profitable_report_permutation}
A census-preserving report permutation $\sigma$ is {\bf profitable} at
$\omega$ if
\[
    u_h\!\left(f^{\mathrm{cc}}(\widehat\omega^\sigma)\mid\omega\right) \geq u_h\!\left(f^{\mathrm{cc}}(\omega)\mid\omega\right),
\]
for every $h\in G$, with a strict inequality for at least one member of $G$.
\end{definition}

The following characterization is useful when the outcome rule treats agents symmetrically. At the truthful profile, let $s_h(\omega)$ denote agent $h$'s \emph{trading position}, that is, her assignment, including the outside option, together with her transfer. Let $U_h(s_k(\omega)\mid\omega)$ be agent $h$'s true payoff from occupying the position assigned to agent $k$.

\begin{definition}
    The outcome rule is {\bf position-equivariant} on side $A$ at $\omega$ if a permutation $\sigma$ of reports assigns each $h\in G$ the position $s_{\sigma(h)}(\omega)$. Thus, for every $h\in G$,
\[
    u_h\!\left(f^{\mathrm{cc}}(\widehat\omega^\sigma)\mid\omega\right) = U_h(s_{\sigma(h)}(\omega)\mid\omega).
\]
\end{definition}

This condition isolates report permutations that exchange complete trading positions. When it is not imposed, Definition~\ref{def:profitable_report_permutation} continues to apply, but the cycle characterization below need not do so.

\begin{proposition}
\label{prop:profitable_permutation_cycles}
Suppose that the count-check outcome rule is position-equivariant on side $A$ at $\omega$. A profitable census-preserving report permutation exists if and only if there are distinct agents, $ h_1,\ldots,h_m\in A$, for some $m\geq 2$, such that, writing $h_{m+1}=h_1$, adjacent agents in the cycle have different types,
\begin{equation}
    t_{h_{\ell+1}}\neq t_{h_\ell} \qquad \text{for every }\ell=1,\ldots,m,
    \label{eq:type_change_report_cycle}
\end{equation}
and
\begin{equation}
    U_{h_\ell}(s_{h_{\ell+1}}(\omega)\mid\omega) \geq U_{h_\ell}(s_{h_\ell}(\omega)\mid\omega) \qquad \text{for every }\ell=1,\ldots,m,
    \label{eq:profitable_report_cycle}
\end{equation}
with at least one inequality in \eqref{eq:profitable_report_cycle} being strict.
\end{proposition}

\begin{proof}
Suppose first that the agents $h_1,\ldots,h_m$ satisfy \eqref{eq:type_change_report_cycle} and \eqref{eq:profitable_report_cycle}. Let $h_\ell$ report the type of $h_{\ell+1}$, with $h_{m+1}=h_1$. Because adjacent agents have different types, every agent in the cycle changes her report. The reports are merely rearranged, so the census is unchanged. Position equivariance assigns $h_\ell$ the position $s_{h_{\ell+1}}(\omega)$. Equation~\eqref{eq:profitable_report_cycle} makes every member weakly better off and at least one member strictly better off. The permutation is profitable.

Conversely, suppose that a profitable census-preserving report permutation exists. Its effective coalition contains only agents whose reports change, so $t_{\sigma(h)}\neq t_h$ for every member $h$. Decompose the permutation into disjoint cycles. Position equivariance and profitability imply that every agent in each cycle weakly prefers the next agent's position to her own. At least one cycle contains a strict improvement. That cycle satisfies both \eqref{eq:type_change_report_cycle} and \eqref{eq:profitable_report_cycle}.
\end{proof}

\begin{definition}
    The positions satisfy {\bf strict self-selection} on side $A$ at $\omega$ if
\begin{equation}
    U_h(s_h(\omega)\mid\omega) > U_h(s_k(\omega)\mid\omega)
    \label{eq:strict_position_self_selection}
\end{equation}
whenever $t_h\neq t_k$. Each agent then strictly prefers the trading position associated with her true report to every position associated with another type.
\end{definition}

\begin{corollary}
\label{cor:strict_self_selection_permutations}
Suppose that the count-check outcome rule is position-equivariant on side $A$ at $\omega$. If strict self-selection \eqref{eq:strict_position_self_selection} holds, no census-preserving report permutation on side $A$ is profitable at $\omega$.
\end{corollary}

\begin{proof}
Every nontrivial report permutation assigns each member of its effective coalition a position associated with a different type. By \eqref{eq:strict_position_self_selection}, each such agent strictly prefers her truthful position. Hence, the permutation cannot satisfy Definition~\ref{def:profitable_report_permutation}.
\end{proof}

Strict self-selection is sufficient but not necessary. Some agents may prefer positions associated with other types without these preferences forming a cycle that benefits all participants. Proposition~\ref{prop:profitable_permutation_cycles} shows that it is the completion of a profitable cycle, rather than any single preference for another position, that matters.

The result concerns coordinated deviations from the truthful equilibrium. It does not imply full implementation. A census-preserving false profile may be unattractive as a joint deviation from truth and nevertheless be stable against unilateral deviations once it has been reached. Appendix~\ref{sec:weak_full_implementation} returns to this distinction.

Definition~\ref{def:profitable_report_permutation} does not allow coalition members to make additional side payments. If such transfers are available, the relevant condition is whether the permutation raises the coalition's total payoff, since any increase can then be redistributed among its members.


\subsection{Large finite markets and the continuum limit}
\label{subsec:large_finite_continuum}

The count-check argument depends on individual reports having positive weight in the market census. This remains true in a large finite market, as changing one report still changes two integer counts. It ceases to be true in a continuum economy, where an individual agent has measure zero.

For every finite $N$ and $M$, Theorem~\ref{thm:main_countcheck} gives a tally with $\max\{K,L\}$ announcements that detects every unilateral change in the census. More generally, Proposition~\ref{prop:residue_tallies} gives, for any fixed coalition size $r$, a tally with at most $(r+1)^{K+L-2}$ announcements that detects every census-changing deviation by a coalition of at most $r$ agents. Neither bound grows with the number of market participants when $K$, $L$, and $r$ are fixed.

This invariance should not be confused with a continuum result. The finite tallies are functions of integer counts, and they remain sensitive to a change of one unit however large the market becomes. A continuum description retains only aggregate measures. It discards changes made on sets of measure zero.

\begin{proposition}
\label{prop:continuum_no_detection}
Suppose that sellers and buyers form atomless measure spaces and that a public signal depends only on the aggregate distributions of reported seller qualities and buyer types. No such signal can detect a false report by one agent through a count check.
\end{proposition}

\begin{proof}
Consider the buyer side. The argument for sellers is the same. Let $(I,\nu)$ be an atomless population of buyers, and let $\theta:I\to \Types$ be a measurable type profile. The associated distribution of buyer types is the measure
\[
        \mu_\theta(A) = \nu\{i\in I \mid \theta(i)\in A\subseteq\Types\}.
\]

Suppose buyer $i_0$ changes her report from $\theta(i_0)$ to another type. The truthful and reported profiles differ only on the singleton $\{i_0\}$. Since the population is atomless, we have that $ \nu(\{i_0\})=0$. The change leaves the distribution $\mu_\theta$ unchanged. A public signal that depends only on this distribution is unchanged as well. The signal calculated from the reported distribution agrees with the signal of the true distribution, so the count check cannot detect the false report.
\end{proof}

The finite and continuum conclusions are compatible. Detection holds at every point in a sequence of finite markets because one participant always changes an integer count. It fails in the nonatomic model because the limiting aggregate distribution assigns no weight to that participant. Passing from finite counts to continuum distributions removes the feature on which the tally relies.

As a result, market tallies should be understood as instruments for finite institutions rather than as approximations to a nonatomic mechanism. Their most natural applications are settings in which participation or inventory remains individually countable, even when the market is large.

\subsection{Disclosure criteria and registry interpretation}
\label{subsec:privacy_objective}

The main analysis measures public disclosure by the number of announcements a tally may produce. For a separating tally $\phi:\Domain\to\mathcal Z$, this number is $|\phi(\Domain)|$. The criterion is prior-free, as every feasible census is considered, regardless of how likely it is to occur.

If the cost of disclosure is an increasing function of the number of possible announcements, the institution weakly prefers a separating tally with the smallest range. On the unrestricted domain, Theorem~\ref{thm:main_countcheck} gives
\[
        \min_{\phi\ \mathrm{separating}} |\phi(\Domain)| = \max\{K,L\}.
\]
The number of announcements is independent of $N$ and $M$ when the sets of seller qualities and buyer types are fixed.

\paragraph{Prior-dependent disclosure.} When a probability distribution over market censuses is available, one may instead measure how much information the public announcement is expected to convey. Let $\mu$ be a probability distribution on $\Domain$. For a tally $\phi:\Domain\to\mathcal Z$, define
\[
        \mu_\phi(z) = \sum_{\{C\in\Domain \,\mid\, \phi(C)=z\}}\mu(C).
\]
This is the probability that the public announcement is $z$. The entropy of the tally, measured in bits, is
\[
        H_\mu(\phi) = - \sum_{\{z\in\mathcal Z \,\mid\,\mu_\phi(z)>0\}} \mu_\phi(z)\log_2\mu_\phi(z).
\]
Because the tally is a deterministic function of the census, this entropy is also the mutual information between the census and its public announcement.

A prior-dependent disclosure policy solves
\begin{equation*}
        \min_{\phi}H_\mu(\phi)
\end{equation*}
subject to $C\sim_1C' \implies \phi(C)\neq\phi(C')$ for every $C,C'\in\Domain$. The constraint preserves the incentive requirement of the count-check mechanism. Even censuses assigned probability zero remain subject to separation if the mechanism is required to work throughout $\Domain$.

The two disclosure criteria need not select the same tally. Minimizing $|\phi(\Domain)|$ controls the number of announcements that may be required. Minimizing $H_\mu(\phi)$ gives greater weight to likely censuses and may favor an uneven distribution of probability across announcements. A tally may then assign one announcement to a large set of likely, mutually nonneighboring censuses while using less frequent announcements for the remainder of the domain.

The entropy criterion may also use more announcements than the smallest possible range. Additional announcements can be assigned to unlikely censuses if doing so permits more probability to be concentrated on one common announcement. Thus, the prior-dependent problem is not generally obtained by first minimizing the number of announcements and then choosing among the minimizers.

\paragraph{Registry interpretation.} A registry provides a natural informational anchor for either disclosure criterion. Suppose that it holds verified records from which it can determine the market census. Before reports are submitted to the trading mechanism, the registry commits to a public rule $\phi:\Domain\to\mathcal Z$ and announces $z=\phi(C)$. Market participants know the rule and can calculate the tally implied by their reports. If $\phi$ is separating, the registry's announcement supports the count check in Theorem~\ref{thm:main_countcheck}.

The registry may need individual records or the full census to produce the announcement. Those data need not be disclosed publicly. What matters is that the registry can authenticate the announcement and that its content is not derived solely from the reports being checked.

A small range does not by itself guarantee that no individual information can be inferred. The implication of an announcement depends on the domain and on what participants already know. The result is best understood as limiting public disclosure for a specified implementation task, and not as providing protection against every form of inference. The same distinction applies to the posted-price institution. An organizer may use detailed information internally to calculate prices, capacities, and clearing instructions while revealing only the certified market structure needed by participants.



\paragraph{Institutional comparison.} The mechanisms constructed in the paper differ in both their public information and their institutional demands. The count-check mechanism can use $\kappa^{\mathrm{cc}}(\Domain)$ public announcements, but it asks agents to report their types and assigns trades directly. The posted-price protocol asks agents to make market choices, but it requires certified prices, seller capacities, and a rule for clearing buyer applications.

These differences do not imply a general ranking. Such a ranking would require a model of the costs of collecting information, certifying it, communicating public announcements, processing reports, and clearing trades. The paper does not provide such a model.

Nor can the informational cost of posted-price clearing always be identified with full public census disclosure. If the public signal is the organizer's only information, Proposition~\ref{prop:seller_capacity_lower_bound} shows that the signal must at least identify seller capacities in the quality-labelled protocol. If the organizer has a private, independently verified record of the census, she may calculate prices and capacities internally and announce only the resulting market structure. The public information required by posted prices depends on what the organizer already knows.

The two mechanisms nevertheless illustrate a useful institutional comparison. Limited public disclosure may be paired with a more active reporting and assignment procedure. A market-choice institution may reduce direct reporting while placing greater demands on certification and coordination. Whether one arrangement is preferable depends on the costs of these activities, not on the number of public announcements alone.

The growing difference between the number of possible censuses and the number of announcements required by a separating tally remains informative about disclosure. It does not, by itself, establish that the count-check mechanism is socially less costly than posted-price clearing. Institutional costs may also vary with market size.


\subsection{Buyer-side composition}
\label{app:buyer_side}

The buyer-side term $L$ in Theorem~\ref{thm:main_countcheck} appears because buyers privately report payoff-relevant types. A false buyer report changes the buyer census, so the public tally must also discipline changes on that side of the market. Buyer counts need not play the same role when buyer types are public or when buyers reveal their information through market choices rather than reports.

\paragraph{A seller-only benchmark.} Suppose that buyer types are publicly known or that there is only one buyer type relevant to the mechanism. Only sellers submit private information. The set of possible seller censuses is
\[
        \mathcal N_N = \left\{ n\in\mathbb Z_+^K \; \text{ such that }\; \sum_{r=0}^{K-1}n_r=N \right\}.
\]
A seller tally is a function $\phi_S:\mathcal N_N\to \mathcal Z$. It separates unilateral seller reports if its announcement changes whenever one seller changes her reported quality.

On the unrestricted seller domain, Theorem~\ref{thm:main_countcheck}, applied with $L=1$, shows that the least number of announcements is $K$. One tally attaining this bound is
\[
        \phi_S(n) = \left( \sum_{r=0}^{K-1}r n_r \right)\bmod K.
\]
The buyer-side term disappears because no private buyer report is being checked.

\paragraph{Private buyer reports.} Now suppose that buyer types are private and are reported to the mechanism. If every buyer census over the $L$ types is feasible, at least $L$ public announcements are needed to reject every unilateral false buyer report.

To see this, fix the seller census and fix $M-1$ buyers at one buyer type. Allow the remaining buyer to have any of the $L$ types. The resulting $L$ censuses are pairwise one-agent neighbors and must receive different announcements. If two shared an announcement, the remaining buyer could change her report between the corresponding types without failing the count check.

Thus, buyer composition matters for the count-check mechanism because buyers submit private reports. It does not follow that every institution must disclose buyer counts.


\subsubsection{Buyer self-selection through prices}
\label{subsec:buyer_self_selection}

Buyer types may remain private even when buyers do not report them. Prices can sometimes induce different types to choose different actions. Fix a seller census $n=(n_q)_{q\in\Qualities}$ and a collection $\mathcal M$ of possible buyer censuses. Suppose there is one price vector $p=(p_q)_{q\in\Qualities}$ that is used for every $m\in\mathcal M$. Let the buyer action set be $\mathcal A=\Qualities\cup\{0\}$, where action $0$ denotes no trade. For a buyer of type $\theta$, define $U_\theta(0)=0$ and $U_\theta(q)=v_\theta(q)-p_q$ for $q\in\Qualities$. Suppose that every buyer type has a unique preferred action. Write $\chi(\theta)$ for that action, so
\begin{equation}
        \argmax_{a\in\mathcal A}U_\theta(a) = \{\chi(\theta)\}.
        \label{eq:unique_buyer_choice}
\end{equation}
For a buyer census $m\in\mathcal M$, the number of buyers who choose quality $q$ is
\[
        d_q(m) = \sum_{\{\theta \,\mid\,\chi(\theta)=q\}}m_\theta.
\]
Assume that these choices respect seller capacity, i.e., 
\begin{equation}
        d_q(m)\leq n_q \quad \text{for every }q\in\Qualities \text{ and every }m\in\mathcal M.
        \label{eq:self_selection_capacity}
\end{equation}

\begin{lemma}
\label{prop:buyer_census_not_needed_self_selection}
Suppose that \eqref{eq:unique_buyer_choice} and \eqref{eq:self_selection_capacity} hold, and that the allocation induced by $\chi$ maximizes total surplus for every $m\in\mathcal M$. Public disclosure of the realized buyer census is not needed to screen buyers. At the posted prices, every buyer chooses the action prescribed by $\chi$, and the resulting allocation is efficient for every $m\in\mathcal M$.
\end{lemma}

\begin{proof}
Equation~\eqref{eq:unique_buyer_choice} gives each buyer a strict best response that depends only on her own type and the posted prices. No buyer report or public buyer count enters that choice. Equation~\eqref{eq:self_selection_capacity} ensures that the resulting demands can be served. By the remaining hypothesis, the resulting allocation maximizes total surplus for every $m\in\mathcal M$.
\end{proof}

The assumptions are demanding. The same prices must screen buyers for every census under consideration, and the resulting choices must never exceed capacity. A simple example illustrates both the possibility and its limitation.

Let $\Qualities=\{q_H,q_L\}$ and $\Types=\{\theta_H,\theta_L\}$, with seller costs $c(q_H)=2$ and $c(q_L)=0$, and buyer values
\[
\begin{array}{c|cc}
        & q_H & q_L \\ \hline
\theta_H & 10 & 4 \\
\theta_L & 3  & 2
\end{array}.
\]
There is one good of each quality, and $\mathcal M$ contains the buyer censuses with at most one buyer of each type. At prices $p_H=6$ and $p_L=1$, a high type obtains payoffs $4$ from $q_H$ and $3$ from $q_L$, and therefore chooses $q_H$. A low type obtains payoffs $-3$ from $q_H$ and $1$ from $q_L$, and therefore chooses $q_L$.

These choices are feasible for every census in $\mathcal M$. They are also efficient, as the high type creates surplus $8$ with $q_H$ and $4$ with $q_L$, while the low type creates surplus $1$ with $q_H$ and $2$ with $q_L$. No buyer report or public buyer census is needed to produce this sorting.

If the domain instead allowed two high-type buyers and only one high-quality good, both buyers would choose $q_H$. Prices would still screen their preferences, but they would not resolve the capacity conflict. This is the reason the general posted-price result retains a clearing rule.

The role of buyer composition is needed when private buyer reports must be checked or when prices and capacities depend on the realized buyer census. It may be unnecessary when one price vector induces efficient and feasible buyer choices throughout the relevant domain.


\subsection{Weak and full implementation}
\label{sec:weak_full_implementation}

Theorem~\ref{thm:main_countcheck} establishes that truthful reporting is an ex post equilibrium. It does not claim that truthful reporting is the only equilibrium. This distinction matters because the count check verifies the composition of reports, not the identity of the agents submitting them.

The profitable-cycle condition in Proposition~\ref{prop:profitable_permutation_cycles} does not settle this question. It asks whether agents would jointly prefer a report permutation to truthful reporting. Full implementation asks instead whether a false report profile can itself be an equilibrium. A profile may fail the first test and still satisfy the second because the agents' incentives to leave it are evaluated after the false profile has been reached.

\begin{definition}
\label{def:weak_full_implementation}
A mechanism \textbf{weakly implements} immediate efficient trade if, at every type profile, it has an equilibrium whose outcome is a surplus-maximizing matching carried out at date zero.

A mechanism \textbf{fully implements} immediate efficient trade if, at every type profile, an equilibrium exists and every equilibrium outcome is a surplus-maximizing matching carried out at date zero.
\end{definition}


The absence of a profitable report permutation does not resolve the question of full implementation. Proposition~\ref{prop:profitable_permutation_cycles} asks whether a coalition would prefer a census-preserving permutation to truthful reporting. Full implementation asks whether a false report profile can itself be an equilibrium. A profile may be unattractive as a deviation from truth and nevertheless be stable once it has been reached.

There are two distinct sources of false equilibria. The first consists of false report profiles that pass the count check. The second consists of rejected profiles from which no agent can restore consistency by changing her report alone.

Fix a true profile $\omega$, and let $C_0=C(\omega)$ and $z_0=\phi(C_0)$, and let $f^{\mathrm{cc}}(\widehat\omega)$ denote the allocation and transfers selected by the count-check mechanism when the report profile $\widehat\omega$ is accepted.

\begin{proposition}
\label{prop:accepted_false_equilibrium}
Suppose that $\phi$ is separating. Let $\widehat\omega$ be any report profile satisfying $C(\widehat\omega)\in\Domain$ and $\phi(C(\widehat\omega))=z_0$. If every agent receives a nonnegative true payoff,
\begin{equation}
    u_h\!\left(f^{\mathrm{cc}}(\widehat\omega)\mid\omega\right) \geq 0 \qquad \text{for every }h\in\Sellers\cup\Buyers,
    \label{eq:accepted_profile_nonnegative}
\end{equation}
then $\widehat\omega$ is a Nash equilibrium of the reporting game following the announcement $z_0$.

As a result, if $\widehat\omega$ is false and its allocation is inefficient at $\omega$, the count-check mechanism does not fully implement immediate efficient trade.
\end{proposition}

\begin{proof}
Fix an agent $h$ and hold all other reports at $\widehat\omega_{-h}$. If agent $h$ submits the same report, the outcome is unchanged. Any different type report either produces a census outside $\Domain$ or changes $C(\widehat\omega)$ to a one-agent neighbor.

In the first case, the report is rejected. In the second, separation gives
\[
    \phi(C(\widetilde\omega_h,\widehat\omega_{-h})) \neq \phi(C(\widehat\omega)) = z_0,
\]
so the report is again rejected. The deviating agent receives zero. By \eqref{eq:accepted_profile_nonnegative}, her payoff at $\widehat\omega$ is nonnegative. No unilateral deviation is profitable, and $\widehat\omega$ is a Nash equilibrium.
\end{proof}

The proposition reveals an unusual feature of the rejection threat. Starting from truth, it discourages a unilateral false report. Starting from an accepted false profile, the same threat may discourage an agent from changing her report. Every accepted profile that gives all agents nonnegative true payoffs is locally stable, whether or not its reports or allocation are correct.

False equilibria need not pass the count check. Let 
\[
    \Domain(z_0) = \{C\in\Domain \mid\phi(C)=z_0\}
\]
be the set of censuses accepted after the announcement $z_0$.

\begin{definition}
\label{def:locally_inescapable_rejection}
A rejected report profile $\widehat\omega$, with reported census $\widehat C=C(\widehat\omega)\in\Domain$, is {\bf locally inescapable} relative to $z_0$ if $\widehat C\notin\Domain(z_0)$ and no one-agent neighbor of $\widehat C$ belongs to $\Domain(z_0)$. Thus, for every $C'\in\Domain$,
\[
    C'\sim_1\widehat C \implies \phi(C')\neq z_0.
\]
\end{definition}

A locally inescapable profile fails the count check, and no agent can make the check pass by changing her report alone.

\begin{lemma}
\label{prop:locally_inescapable_equilibrium}
Every locally inescapable rejected report profile is a no-trade Nash equilibrium of the reporting game following the announcement $z_0$.
\end{lemma}

\begin{proof}
At the rejected profile, the mechanism cancels trade and every agent receives zero. A unilateral change either produces a census outside $\Domain$ or moves the reported census to a one-agent neighbor. In the first case, the new report is rejected by feasibility. In the second, it is rejected by Definition~\ref{def:locally_inescapable_rejection}. Every unilateral deviation also gives zero. No agent can obtain a strict improvement, so the rejected report profile is a Nash equilibrium.
\end{proof}

The problem is particularly transparent when the full census is public. Recall the census distance $d$ defined in Appendix~\ref{app:coalitions}.

\begin{corollary}
\label{cor:full_census_rejection_equilibrium}
Suppose that $\phi$ reveals the full census. If a report profile $\widehat\omega$ satisfies $d(C(\widehat\omega),C_0)\geq 2$, then $\widehat\omega$ is a no-trade Nash equilibrium following the announcement of the true census. If $W^*(\omega)>0$, this equilibrium is inefficient, and the mechanism does not fully implement immediate efficient trade.
\end{corollary}

\begin{proof}
Under full census disclosure, the only accepted census is $C_0$. A unilateral change from $C(\widehat\omega)$ produces a census at distance at least one from $C_0$, because a single report change can reduce the distance by at most one. Hence, no unilateral change can produce the accepted census. The profile is locally inescapable, and Lemma~\ref{prop:locally_inescapable_equilibrium} applies. The equilibrium outcome is no trade. If $W^*(\omega)>0$, no trade is inefficient.
\end{proof}

Proposition~\ref{prop:accepted_false_equilibrium} and Lemma~\ref{prop:locally_inescapable_equilibrium} identify different coordination failures. In the first, the reports pass the count check but produce the wrong outcome. In the second, the reports fail the check, but no individual agent can move the profile into the accepted set. Full implementation requires ruling out both possibilities.

The next example illustrates Proposition~\ref{prop:accepted_false_equilibrium}. The false reports preserve the census, pass the count check, and give every agent a positive payoff, but the resulting allocation is inefficient.

\begin{proposition}
\label{prop:count_check_not_full}
Even when the full market census is public, the count-check mechanism may have an inefficient equilibrium.
\end{proposition}

\begin{proof}
There are two sellers and two buyers. Seller qualities are $\Qualities=\{q_H,q_L\}$, buyer types are $\Types=\{\theta_H,\theta_L\}$, and seller costs are $c(q_H)=c(q_L)=0$. Buyer values are
\[
\begin{array}{c|cc}
        & q_H & q_L \\ \hline
\theta_H & 10 & 5 \\
\theta_L & 6  & 4
\end{array}.
\]

At the realized profile, seller 1 owns the high-quality good, seller 2 owns the low-quality good, buyer 1 has type $\theta_H$, and buyer 2 has type $\theta_L$. The efficient matching pairs buyer 1 with seller 1 and buyer 2 with seller 2, producing surplus $10+4=14$.

Suppose the public tally reveals the full census. Sellers report their qualities truthfully, but the buyers exchange reports $\hat\theta_1=\theta_L$ and $\hat\theta_2=\theta_H$. The reported buyer census remains one high type and one low type, so the count check passes.

Given the reports, the mechanism assigns the reported high type to the high-quality good and the reported low type to the low-quality good. Buyer 2 is matched with seller 1, and buyer 1 is matched with seller 2.

The midpoint price paid by buyer 2 is
\[
        \frac{v_{\theta_H}(q_H)+c(q_H)}{2}=5.
\]
Since her true type is $\theta_L$, her payoff is $v_{\theta_L}(q_H)-5=1$. Buyer 1 pays
\[
        \frac{v_{\theta_L}(q_L)+c(q_L)}{2}=2
\]
and receives the true payoff $v_{\theta_H}(q_L)-2=3$. Seller 1 receives $5$, and seller 2 receives $2$. Every agent obtains a positive payoff.

The exchanged-report profile passes the full-census check, and every agent receives a strictly positive true payoff. It is a Nash equilibrium by Proposition~\ref{prop:accepted_false_equilibrium}. Its allocation produces surplus $v_{\theta_L}(q_H)+v_{\theta_H}(q_L) = 6+5=11$, which is below the first-best surplus of $14$. The count-check mechanism does not fully implement immediate efficient trade.
\end{proof}

The example in the proof of Proposition~\ref{prop:count_check_not_full} illustrates the difference between verifying counts and verifying identities. The full census certifies that there is one buyer of each type, but it does not link those types to buyer identities. Each buyer can combine the census with her own private type and infer the other buyer's type. This private inference does not provide the mechanism with verifiable evidence about which buyer is the high type. Truthful reports and exchanged reports imply the same census, so the aggregate check accepts both.

The inefficient equilibrium need not be a profitable joint deviation from the truthful equilibrium. In the example, both buyers receive less under the exchange than under truthful reporting. Full implementation nevertheless fails because the inefficient report profile is itself stable against unilateral deviations. Full implementation concerns all equilibria, not only deviations that agents would jointly choose from the truthful equilibrium.

The rejection equilibria in Lemma~\ref{prop:locally_inescapable_equilibrium} make the distinction still stronger. They need not arise from a coordinated deviation that anyone finds attractive. They are equilibria because the agents face a coordination problem. That is, no one can make the count check pass alone. Thus, strengthening the truthful equilibrium against profitable coalitional deviations does not by itself eliminate inefficient equilibria elsewhere in the reporting game.

Eliminating such equilibria requires an additional source of discipline. Identity-linked certification can distinguish the two buyers' reports. Audits or legal penalties can make individual false reports unattractive even when the census is preserved. A different allocation or transfer rule may also remove the gains that stabilize a false-report profile.

If audits make truthful reporting strictly dominant, full implementation may be recovered, but the audit then performs the main incentive function. The contribution of the tally result is more limited. Without individual verification, a small certified statistic can support truthful reporting as an ex post equilibrium.


\subsection{Decentralized posted prices with strict sorting}
\label{app:decentralized_prices}

The posted-price protocol in Section~\ref{sec:posted_prices} uses an organizer to allocate buyers among markets. This subsection gives sufficient conditions under which applications can be replaced by direct buyer choice. The conditions are quite demanding, as prices must separate buyer types strictly, buyer choices must respect capacity, and sellers must not gain by entering price classes intended for other qualities.

Fix a public census $C= \left( (n_q)_{q\in\Qualities}, (m_\theta)_{\theta\in\Types} \right)$, and let
\[
        \Qualities_C = \{q\in\Qualities \mid n_q>0\},
\]
\[
        \Types_C = \{\theta\in\Types \mid m_\theta>0\}.
\]

Consider the following seller-posting game. A finite menu
\[
        P_C\coloneqq\{p_q \mid q\in\Qualities_C\}
\]
is publicly announced, with $p_q\neq p_{q'}$ whenever $q\neq q'$. Each seller posts one price from this menu. Buyers observe all posted prices and seller identities, and then choose one seller or the outside option. A seller chosen by one buyer trades with that buyer at the posted price. If several buyers choose the same seller, one trades according to a fixed public tie-breaking rule, and the others remain unmatched.

The distinct prices identify the price classes, but they do not directly verify the quality of an individual seller. Beliefs following an unexpected change in the number of sellers posting a price matter for seller incentives.

\begin{definition}
\label{def:strict_decentralized_support}
A census $C$ admits a \textbf{strict decentralized price support} if there are
\begin{enumerate}
    \item a set of trading buyer types $T\subseteq\Types_C$;

    \item an assignment $\chi:T\to \Qualities_C$;

    \item distinct prices $(p_q)_{q\in\Qualities_C}$,
\end{enumerate}
such that the following conditions hold:
\begin{enumerate}[label=\textup{(\roman*)},leftmargin=2em]
    \item every trading type is represented by one buyer. That is, $m_\theta=1$ for every $\theta\in T$;

    \item matching the buyer of type $\theta\in T$ with one good of quality $\chi(\theta)$, and leaving all other buyers unmatched, maximizes total surplus;

    \item the assignment respects capacity, i.e., $\#\{\theta\in T \mid \chi(\theta)=q\} \leq n_q$ for every $q\in\Qualities_C$;

    \item every trading type receives a positive payoff from her assigned quality-price pair and strictly prefers it to every other active quality. That is, for every $\theta\in T$,
\begin{equation}
\begin{aligned}
        v_\theta(\chi(\theta)) - p_{\chi(\theta)}
        &>0,\\
        v_\theta(\chi(\theta)) - p_{\chi(\theta)}
        &> v_\theta(q)-p_q \qquad \text{for every $q\in\Qualities_C$ with } q\neq\chi(\theta);
\end{aligned}
\label{eq:strict_self_selection}
\end{equation}

    \item every nontrading type strictly prefers the outside option, i.e., for every $\theta\in\Types_C\setminus T$,
    \begin{equation}
            0 > \max_{q\in\Qualities_C} [v_\theta(q)-p_q];
            \label{eq:strict_outside}
    \end{equation}

    \item every seller receives a nonnegative payoff from selling at the price intended for her quality, i.e., $p_q\geq c(q)$ for every $q\in\Qualities_C$. The inequality is strict whenever some trading buyer type is assigned to quality $q$;

    \item consider qualities $r,q\in\Qualities_C$ with $r\neq q$. If one seller of quality $r$ posts the price $p_q$ and $p_q>c(r)$, then
    \begin{equation}
            \frac{ n_qv_\theta(q)+v_\theta(r) }{n_q+1} - p_q \leq0
            \label{eq:mimicking_price_class}
    \end{equation}
    for every $\theta\in\Types_C$.
\end{enumerate}
\end{definition}

Condition~\eqref{eq:strict_self_selection} assigns each trading buyer type to one price class. Condition~\eqref{eq:strict_outside} keeps all other buyers out of the market. Capacity ensures that the trading types can be directed to different sellers within each class.

The final condition addresses adverse selection by sellers. If one seller of quality $r$ enters the price class intended for quality $q$, buyers observe one additional seller at $p_q$ and one missing seller at $p_r$. Under a belief that the deviator is equally likely to be any seller in the enlarged class, the left-hand side of \eqref{eq:mimicking_price_class} is a type-$\theta$ buyer's expected payoff from choosing a seller in that class. The condition makes the entire class unattractive after the deviation.

\begin{theorem}
\label{thm:decentralized_price_implementation}
Suppose that the public census $C$ admits a strict decentralized price support. The seller-posting game has an ex post equilibrium in which every seller of quality $q$ posts $p_q$, every trading buyer type purchases from a distinct seller in the class associated with $\chi(\theta)$, and every nontrading buyer chooses the outside option. The resulting allocation is efficient and occurs at date zero.

For any common prior over identity assignments consistent with $C$, the strategies can be completed with beliefs to form a pure-strategy perfect Bayesian equilibrium.
\end{theorem}

\begin{proof}
For each quality $q$, let
\[
        T_q \coloneqq \{\theta\in T \mid \chi(\theta)=q\}.
\]
By capacity, $|T_q|\leq n_q$. Fix a public ordering of seller identities. Within each price class, rank the sellers according to that ordering. For every $q$, assign the types in $T_q$ injectively to the first $|T_q|$ ranks in the price-$p_q$ class. Denote the rank assigned to type $\theta$ by $R(\theta)$.

Consider the following strategies. A seller of quality $q$ posts $p_q$. Call a seller-price profile \emph{regular} if $n_q$ sellers post $p_q$ for every $q\in\Qualities_C$. At a regular profile, the buyer of type $\theta\in T$ chooses the seller with rank $R(\theta)$ in the price class $p_{\chi(\theta)}$. A buyer whose type does not belong to $T$ chooses the outside option.

At a history produced by one seller of quality $r$ changing her price from $p_r$ to $p_q$, buyers believe that the enlarged price-$p_q$ class contains $n_q$ sellers of quality $q$ and one seller of quality $r$. They assign equal probability to each seller in that class being the deviator. Buyers choose best responses under these beliefs and, when choosing the enlarged class and the outside option both give zero, select the outside option. At all other off-path histories, beliefs may be completed together with buyer best responses.

On the equilibrium path, the seller-price profile is regular. The public rank rule sends different trading buyer types to different sellers. There is no congestion, and condition (ii) of Definition~\ref{def:strict_decentralized_support} implies that the resulting allocation is efficient.

Consider buyer incentives. A buyer of type $\theta\in T$ obtains $v_\theta(\chi(\theta)) - p_{\chi(\theta)}$. By \eqref{eq:strict_self_selection}, this payoff is positive and exceeds the payoff from every other quality-price pair. Choosing another seller in the same price class gives the same payoff if trade occurs and may instead create competition for that seller, so it cannot give a higher payoff.

If $\theta\notin T$, equation~\eqref{eq:strict_outside} implies that purchasing from any seller gives a negative payoff. Therefore, outside option is the buyer’s strict best response.

Now consider a seller of quality $r$. Posting $p_r$ gives her a nonnegative payoff. Suppose she posts $p_q$ for some $q\neq r$. If $p_q\leq c(r)$, then even a sale at $p_q$ cannot give her a positive payoff, so the deviation is not profitable.

If $p_q>c(r)$, the price-$p_q$ class contains one seller more than its public capacity. Under the specified beliefs, a type-$\theta$ buyer's expected payoff from choosing a seller in that class is
\[
        \frac{ n_qv_\theta(q)+v_\theta(r) }{n_q+1} - p_q,
\]
which is nonpositive by \eqref{eq:mimicking_price_class}. Buyers strictly prefer another action when the expression is negative and choose the outside option when it is zero. No buyer enters the enlarged price class. The deviating seller receives zero, which does not exceed her equilibrium payoff.

The prescribed strategies are mutual best responses at every realized identity assignment consistent with $C$, so they form an ex post equilibrium.

Finally, fix a common prior over those identity assignments. On the equilibrium path, seller prices reveal the qualities prescribed by the strategy, and beliefs follow from Bayes' rule. A one-seller deviation has zero probability under the equilibrium strategies, so the beliefs specified above may be used at that history. Beliefs and buyer best responses can be assigned at the remaining zero-probability histories. The resulting assessment is a pure-strategy perfect Bayesian equilibrium.
\end{proof}

The theorem identifies a case in which clearing can be replaced by a public coordination convention. The convention does not require the organizer to observe buyer types. Each buyer uses her own type to identify a price class and a seller rank. It does, however, rely on the public census, distinct prices, a common ordering of sellers, and no duplication among the buyer types that trade.

These assumptions explain why the main result uses an organizer. If several buyers of the same trading type are present, the type-contingent rank rule no longer distinguishes them. If preferences are not strict, buyers may choose different markets from those intended. If demand exceeds capacity, some buyers must be reassigned. And if an additional seller does not make an enlarged price class unattractive, a seller may profit by imitating another quality. The clearing protocol handles these cases without imposing the restrictions used here.

	\clearpage
	\addcontentsline{toc}{section}{References}
	\bibliographystyle{chicago}
	\bibliography{biblio_tallies.bib} 
\end{document}